\documentclass[journal,twoside,web]{ieeecolor}
\usepackage{generic}
\usepackage{cite}
\usepackage{amsmath,amssymb,amsfonts}
\usepackage{algorithmicx,algorithm}
\usepackage{graphicx}
\usepackage{textcomp}

\newtheorem{theorem}{Theorem}
\newtheorem{lemma}{Lemma}
\newtheorem{remark}{Remark}
\newtheorem{corollary}{Corollary}
\newtheorem{proposition}{Proposition}
\newtheorem{example}{Example}
\newtheorem{assumption}{Assumption}
\def\bb{\mathbb}
\def\det{\mathop{\rm det}}
\def\diag{\mathop{\rm diag}}

\def\max{\mathop{\rm max}}
\def\min{\mathop{\rm min}}

\def\BibTeX{{\rm B\kern-.05em{\sc i\kern-.025em b}\kern-.08em
    T\kern-.1667em\lower.7ex\hbox{E}\kern-.125emX}}
\begin{document}
\title{Fully Distributed State Estimation for Multi-agent Systems and its Application in Cooperative Localization}
\author{Shuaiting Huang, Haodong Jiang, Chengcheng Zhao, Peng Cheng, Junfeng Wu* 
\thanks{S. Huang, C. Zhao, and P. Cheng are with the College of Control Science and Engineering, Zhejiang University, Hangzhou 310027, China (e-mails: \{shuait\_huang, chengchengzhao, lunarheart\}@zju.edu.cn). H. Jiang and J. Wu are with the School of Data Science, the Chinese University of Hong Kong, Shenzhen, China (email: haodongjiang@link.cuhk.edu.cn, junfengwu@cuhk.edu.cn).}}

\markboth{Journal of \LaTeX\ Class Files,~Vol.~18, No.~9, September~2020}%
{How to Use the IEEEtran \LaTeX \ Templates}

\maketitle

\begin{abstract}
In this paper, we investigate the distributed state estimation problem for a continuous-time linear multi-agent system (MAS) composed of $\mathit{m}$ agents and monitored by the agents themselves. To address this problem, 
we propose a distributed observer that enables each agent to reconstruct the state of the MAS.
The main idea is to let each agent $\mathit{i}$ recover the state of agent $\mathit{j}$ by using leader-follower consensus rules to track agent $\mathit{j}$'s state estimate, which is generated by agent $\mathit{j}$ itself using a Luenberger-like estimation rule. 
Under the assumptions of node-level observability and topological ordering consistency,
we show that the estimation error dynamics are
stabilizable if and only if the communication graph is strongly connected.
Moreover, we discuss the fully distributed design of the proposed observer, assuming that the agents only know basic MAS configuration information, such as the homogeneity and the maximum number of allowable agents. 
This design ensures that the proposed observer functions correctly when agents are added or removed. 
Building on this, we consider cooperative localization as a distributed estimation problem and develop two fully distributed localization algorithms that allow agents to track their own and other agents' positions (and velocities) within the MAS.
Finally, we conduct simulations to demonstrate the effectiveness of our proposed theoretical results.
\end{abstract}


\section{Introduction}
\label{sec:introduction}
\IEEEPARstart{D}{istributed} state estimation (DSE) is a fundamental focus of multi-agent systems (MASs) due to its broad applications in domains such as robotics, autonomous vehicles, smart grid, and sensor networks~\cite{monticelli2000electric,ghabcheloo2009coordinated,soares2015simple}.
% The DSE problem for MASs refers to the design of a set of state estimate updates and information exchange rules, namely distributed estimator, to enable each agent to reconstruct the MAS state based only on its local measurement and communication. 
Over the past few decades, a significant amount of research has been conducted on DSE \cite{huang2022distributed}.
% which can be categorized into two groups depending on the underlying framework: Kalman-filter-based approaches and Luenberger-observer-based approaches \cite{huang2022distributed}.
Kalman-filter-based approaches, known as distributed Kalman filters (DKFs), aim to acquire optimal estimates by minimizing estimation error covariance in the presence of Gaussian noise. 
In particular, Rao and Durrant-Whyte proposed a distributed implementation of the centralized Kalman filter in \cite{rao1991fully}, initiating the study of distributed estimation.
Following this, consensus-based DKFs emerged as a representative class of DKFs.
Olfati-Saber established the benchmark for consensus-based DKFs in \cite{olfati2005distributed}, with their stability and optimality later analyzed in \cite{khan2011stability}.
These DKFs operate on a two-time-scale framework, where multiple consensus iterations occur between consecutive estimate updates.
Subsequent works \cite{battistelli2014consensus,marelli2021distributed,hao2022distributed} proposed single-time-scale DKFs, where data fusion and estimate updates are performed simultaneously.

On the other hand, distributed observers (DOs), or Luenberger-observer-based approaches, focus on ensuring the stability of the estimation error dynamics for the systems without precise statistical models of the process and measurement noise.
Noise is not a primary concern here, and researchers typically consider designing DOs for noiseless systems or systems with bounded noise/disturbance.
For example, Rego \emph{et al.} proposed a consensus-based distributed Luenberger observer for systems with bounded noise in \cite{rego2016design}.
Park and Martins applied decentralized control theory to develop a distributed observer \cite{park2016design}.
Romeo \emph{et al.} reformulated the DSE problem as parameter estimation and proposed a distributed observer with finite convergence time \cite{ortega2020distributed}.  
Additionally, several DOs were designed based on observability/detectability decomposition \cite{mitra2018distributed,kim2019completely,del2019distributed}.
However, most existing distributed estimation algorithms are inherently ``non-resilient'' to node changes in the MAS (e.g., agents joining or leaving), as their construction relies on global information about the MAS, including knowledge of the  communication graph topology and the output matrices of other agents.
In practice, such information is often unavailable to individual agents, and node changes in MASs are common due to component failures, attacks, communication losses, agents that temporarily sleep to save energy or task reassignments.
Few existing distributed estimation algorithms resilient to such MAS changes have certain limitations. 
For example, the plug-and-play distributed estimator in \cite{riverso2013plug} and \cite{siljak2013decentralized} enables each agent to estimate only its own state, while the plug-and-play distributed observer in \cite{yang2023plug} requires an undirected communication graph. 
To overcome these limitations, this paper proposes a fully distributed observer that enables each agent to estimate the entire state of the MAS and functions correctly when facing the changes of the MAS under relatively mild assumptions.

On the other hand, cooperative localization, as the premise for MASs to enable high-level collective tasks autonomously, has attracted much attention by developing different protocols (see, for instance,\cite{oh2012formation,oh2013formation,zhu2019multi,roumeliotis2002distributed,chang2021resilient}, and references therein).
To mention a few, Oh and Ahn developed a distributed position estimation law for agents to obtain their own positions using relative position measurements \cite{oh2012formation}.
They also proposed a localization algorithm 
with distance and angle measurement sensors under  
local reference frame alignment~\cite{oh2013formation}.
Two extended-Kalman-filter-based approaches are respectively proposed to ensure that each agent obtains the absolute positions of all agents \cite{roumeliotis2002distributed,chang2021resilient}. 
However, all these localization algorithms are distributed in operation but centralized in construction, making them non-resilient to changes in MASs. 
To address this, this paper provides a new fully distributed solution to the cooperative localization problem by formulating it as a DSE problem and applying the proposed distributed estimation framework.

Motivated by the above analysis, we study a fully distributed state estimation problem for the MASs with couplings among agents and use a widely studied class of multi-agent cooperative localization problems to highlight the effectiveness of the proposed results.
By imposing mild constraints on the MASs, this paper develops novel fully distributed observer and localization algorithms.
The main contributions are summarized as follows.
For the MASs with the state-to-state coupling and the state-to-output coupling between different agents,
i) we propose a distributed observer that enable each agent to reconstruct the entire MAS state, and establish necessary and sufficient conditions for stabilizing the estimation error dynamics under both full and partial input accessibility;
ii) to ensure resilience to node changes of the MAS, we discuss how the proposed observer can be designed in a fully distributed manner over both undirected and directed communication graphs.
For the cooperative localization problem of the MAS comprising the well-known single- and double-integrator modeled agents,
iii) we provide a rigorous observability analysis and develop a distributed directed acyclic graph construction algorithm to address the couplings among agents caused by relative position measurements; 
iv) we propose two fully distributed localization algorithms for single- and double-integrator MASs based on our distributed estimation framework. These algorithms enable each agent to reconstruct all agents' positions (and/or velocities) while admitting a fully distributed design.

The rest of the paper is organized as follows.
In Section~\ref{sec2}, we formulate the DSE problem of our interest.
In Section~\ref{sec3}, we introduce the proposed distributed observer, give a stabilizability analysis and 
discuss the fully distributed design of the observer.
In Section~\ref{sec4}, we discuss the extension of our framework to handle cooperative localization problems. 
% Two fully distributed localization strategies are proposed for MASs consisting of single- and double-integrator modeled agents, respectively.
We present simulation results in Section~\ref{sec5} and provide concluding remarks in Section~\ref{sec6}.

% \noindent\textbf{Notations}.
% Let $\bb{R}$ and $\bb{N}^{+}$ denote the set of real numbers and the set of positive integers. 
% Let $ {A}^\top$ represent the transpose of a matrix $ {A}$.
% The matrix $ {I}_m$ denotes the $m$-dimensional identity matrix.
% For matrices $X_1,\ldots,X_m$, $\textup{diag}(X_1,\ldots,X_m)$ denotes a block diagonal matrix whose $i$th diagonal block is $X_i$.
% The symbol $\mathbf{1}_m$  denotes an $m\times 1$ all-ones column vector.
% The symbol $b_i$ denotes a vector that has $1$
%  as the $i$th
% component and zeros elsewhere.
% The symbol $\mathbf{0}$ denotes a zero matrix whose dimensions are clear from the context.
% The symbol $\otimes$ denotes a Kronecker product.

\section{Preliminaries and Problem Formulation}\label{sec2}
\subsection{Preliminaries of Graph Theory}\label{sec21}
A weighted directed graph, denoted by $\mathrm G=(\mathcal V, \mathcal E, W)$, is comprised of a finite nonempty node set $\mathcal V=\{1,\ldots, m\}$, an edge set $\mathcal E\subseteq\mathcal V\times\mathcal V$, and a weighted adjacency matrix $W=[w_{ij}]\in \bb{R}^{m\times m}$.
An edge from node $j$ to node $i$ in set $\mathcal V$, denoted by $(j,i)\in \mathcal V$, implies that node $j$ can transmit information to node $i$.
The weight $w_{ij}$ is assigned to
$(j,i)\in\mathcal E$, where $w_{ij}>0$ if and only if $(j,i)\in\mathcal E$, and $w_{ij}=0$ otherwise.
The in-neighbor set and out-neighbor set of node $i$ are defined as $\mathcal N_i=\{j\in\mathcal V/\{i\}|(j,i)\in\mathcal E\}$ and $\mathcal S_i=\{j\in\mathcal V/\{i\}|(i,j)\in\mathcal E\}$, respectively.
The in-degree of node $i$ is the number of in-neighbors of $i$.
Node $i$ is called a source if its in-degree equals $0$.
A directed path from node $i_0$ to node $i_l$ is an ordered sequence of edges $(i_{k-1},i_{k})$, $k\in\{1,\ldots, l\}$. 
A directed graph $\mathrm G$ is strongly connected if there exists a directed path between an arbitrary pair of distinct nodes $i$ and $j$ in set $\mathcal V$.
The Laplacian matrix ${L}=[l_{ij}]\in \bb{R}^{m\times m}$ of graph $\mathrm G$ is defined as $ {L}= {D}-{W}$, where $ {D}$ is a diagonal matrix with the $i$th diagonal element being $d_i=\sum_{j=1}^{m}w_{ij}$.
A topological ordering of graph $\mathrm G$ is a sequence of all the nodes such that, for every directed edge, its starting node is listed earlier in the ordering than the ending node.
% Given a matrix $ {W}=[w_{ij}]\in \mathbb{R}^{m\times m}$ satisfying $w_{ii}=0$ and $w_{ij}\ge0$, $i,j=1, \ldots, m$, $i\neq j$, we can always define a directed graph $\mathrm G$, called the directed graph of $W$, whose weighted adjacency matrix is $W$.
We introduce the following lemma.
\begin{lemma}
\label{le:spectrum of L}
(Lemma 3.3 in \cite{ren2005consensus})
Given a weighted directed graph $\mathrm G$, the Laplacian matrix $L$ has at least one zero eigenvalue and all nonzero eigenvalues have strictly-positive real parts.
Furthermore, $L$ has exactly one zero eigenvalue if and only if $\mathrm G$ has a spanning tree.
\end{lemma}
% A directed graph $\mathrm G_p=(\mathcal V_p, \mathcal E_p)$ is called a subgraph of graph $\mathrm G$ if it is satisfied that $\mathcal V_p\subseteq \mathcal V$ and $\mathcal E_p\subseteq\mathcal E$.
% A subgraph $\mathrm G_p=(\mathcal V_p, \mathcal E_p)$ of graph $\mathrm G$ is called a spanning tree of graph $\mathrm G$ if graph $\mathrm G_p$ is a directed tree and $\mathcal V_p=\mathcal V$. 
% A directed graph $\mathrm G$ has a spanning tree if a spanning tree is a subgraph of graph $\mathrm G$.
% A directed graph $\mathrm G$ has a spanning tree if and only if graph $\mathrm G$ has at least one node which can reach all other nodes.
% A directed acyclic graph is a directed graph whose nodes and edges will never form a closed loop. 

\subsection{System Model}\label{sec22}

Consider a multi-agent system consisting of $m$ agents modeled by the following continuous-time linear time-invariant systems:

\begin{small}
\begin{align}
&\dot{x}_i= {A}_{ii}x_i+\sum_{j\in\mathcal N_i(\mathrm{G}_{\rm s})}A_{ij}x_j+B_{ii}u_i \label{def:MAS_dynamics}\\
& y_i=C_{ii}x_i+\sum_{j\in\mathcal N_i(\mathrm{G}_{\rm o})}C_{ij}x_j,~i\in\{1,\ldots, m\},\label{def:MAS_observation}
\end{align}
\end{small}where $x_i\in \mathbb{R}^{n_i}$, $u_i\in \mathbb{R}^{k_i}$, ${y}_i\in \mathbb{R}^{p_i}$, $A_{ii}\in\mathbb{R}^{n_i\times n_i}$, $B_{ii}\in\mathbb{R}^{n_i\times k_i}$ and $C_{ii}\in\mathbb{R}^{p_i\times n_i}$ are the state, input, output, system matrix, input matrix and observation matrix of agent $i$, respectively.
The matrices ${A}_{ij}\in\mathbb{R}^{n_i\times n_j}$ and $ {C}_{ij}\in\mathbb{R}^{p_i\times n_j}$ are the state- and output-coupling matrices between agents $i$ and $j$, respectively, describing the ``state-to-state'' and ``state-to-output'' relationships. 
These relationships are inspired by real-world systems, including vehicle platooning and water distribution systems.
We take the vehicle platooning system as an example.
In a vehicle platooning system, the physical coupling induced by wind resistance between vehicles, as discussed in \cite{al2011suboptimal}, exemplifies a ``state-to-state'' relationship.
Likewise, the radar-measured intervehicle distance between a vehicle and its preceding counterpart in a platoon can be modeled as a ``state-to-output'' relationship \cite{huang2023plug}.

The dynamics graph $\mathrm{G}_{\rm s}=(\mathcal V, \mathcal E_{\rm s})$ and the sensing graph $\mathrm{G}_{\rm o}=(\mathcal V, \mathcal E_{\rm o})$ are introduced to characterize the ``state-to-state'' and ``state-to-output'' relationships, respectively.
The node set $\mathcal V=\{1,\ldots, m\}$ is the collection of the agent indices, with the edge sets $\mathcal E_{\rm s}, \mathcal E_{\rm o} \subseteq\mathcal V\times\mathcal V$. 
Edges $(j, i) \in \mathcal{E}_{\rm s}$ and $(j, i) \in \mathcal E_{\rm o}$ if and only if ${A}_{ij}\neq  \mathbf{0}$ and $ {C}_{ij}\neq \mathbf{0}$, respectively, indicating agent $j$'s state affects the state and output of agent $i$.
The symbols $\mathcal N_i(\mathrm{G}_{\rm s})$ and $\mathcal N_i(\mathrm{G}_{\rm o})$ denote the set of in-neighbors of agent $i$ in the graphs $\mathrm{G}_{\rm s}$ and $\mathrm{G}_{\rm o}$, respectively.

By stacking the states of $m$ agents, we obtain the overall dynamics of the MAS as follows:

\begin{small}
\begin{align}
&\dot{x}= Ax+Bu \label{eq:sys1}  \\
& y= Cx \label{eq:sys2}
\end{align}
\end{small}where $x=[{x}_1^\top, \ldots,  {x}_m^\top]^\top\in\mathbb{R}^{n}$, $u=[{u}_1^\top, \ldots,  {u}_m^\top]^\top\in\mathbb{R}^{k}$, $y=[{y}_1^\top, \ldots,  {y}_m^\top]^\top\in\mathbb{R}^{p}$, $A=[A_{ij}]\in\mathbb{R}^{n\times n}$, $B=\textup{diag}(B_{11},\ldots, B_{mm}) \in\mathbb{R}^{n\times k}$, $C=[C_{ij}]\in\mathbb{R}^{p\times n}$.
The communication graph of the MAS \eqref{eq:sys1}-\eqref{eq:sys2}, denoted by $\mathrm{G}_{\rm c}=(\mathcal V, \mathcal E_{\rm c})$, is introduced to model how the agents exchange information with their neighbors through a communication network, where the edge set $\mathcal E_{\rm c}\subseteq\mathcal V\times\mathcal V$ is defined so that $(j, i)\in \mathcal E_{\rm c}$ indicates the capability of agent $j$ to send information to agent $i$.

\subsection{Problem Formulation}
\label{sec23}
The objective of this paper is to design a fully distributed observer structure and conduct a stability analysis of the resulting estimation error dynamics under full and partial input information. 
The proposed observer must enable each agent $i\in\mathcal V$ to achieve stable estimation of the entire state $x$ of the MAS \eqref{eq:sys1}-\eqref{eq:sys2}, using only local measurement~\eqref{def:MAS_observation} and information transmitted by neighbors via the communication graph $\mathrm{G}_{\rm c}$.
Furthermore, we aim to demonstrate that the observer can be constructed and implemented in a fully distributed manner.
To achieve this, we impose the following assumptions.

\begin{assumption}[Node-level Observability]
\label{ass:local observability}
For each agent $i\in\mathcal V$, the matrix pair $({A}_{ii},  {C}_{ii})$ is observable.
\end{assumption}

\begin{assumption}[Topological Ordering Consistency]
\label{ass:Gs-Go}
The dynamics graph $\mathrm G_{\rm s}$ and the sensing graph $\mathrm G_{\rm o}$ are directed acyclic graphs (DAGs) with the same topological ordering.
\end{assumption}

\begin{remark}
Assumption~\ref{ass:local observability} does not imply that agent $i$ can estimate its own state solely from its output $y_i$, due to inter-agent coupling.
This assumption is stronger than the joint observability assumption in~\cite{wang2017distributed,kim2019completely} (To be precise, Assumptions~\ref{ass:local observability}-\ref{ass:Gs-Go} ensure that the MAS \eqref{eq:sys1}-\eqref{eq:sys2} is jointly observable, a necessary condition for solving the DSE problem).
However, Assumptions~\ref{ass:local observability}-\ref{ass:Gs-Go} support the design of a fully distributed observer for coupled MASs; see Section~\ref{sec3} for greater detail.
Furthermore, these two assumptions are satisfied by certain practical MASs, such as the vehicle platooning system \cite{al2011suboptimal,huang2023plug} and the MAS studied in the cooperative localization problem of our interest.
\end{remark}

We give the following example of a MAS satisfying Assumption~\ref{ass:Gs-Go}.

\begin{example}\label{exam:assumption3}
Consider a MAS consisting of four agents whose dynamics are presented in \eqref{eq:example1}.
The corresponding graphs $\mathrm G_{\rm s}$ and $\mathrm G_{\rm o}$ are illustrated in Fig.~\ref{fig:example1}.
It can be seen that the system \eqref{eq:example1} satisfies Assumption~\ref{ass:Gs-Go}, of which the topological ordering is $4,~3,~2,~1$.

\begin{footnotesize}
\begin{subequations}
\label{eq:example1}
\begin{align}
& {\rm Agent\,1:}\left\{
   \begin{aligned}
   & \dot{x}_1=A_{11}x_1+A_{12}x_2+A_{13}x_3\\
   & y_1=C_{11}x_1+C_{12}x_2
   \end{aligned}
\right.\\
& {\rm Agent\,2:}\left\{
   \begin{aligned}
   & \dot{x}_2=A_{22}x_2+A_{23}x_3\\
   & y_2=C_{22}x_2+C_{23}x_3
   \end{aligned}
\right.\\
& {\rm Agent\,3:}\left\{
   \begin{aligned}
   & \dot{x}_3=A_{33}x_3+A_{34}x_4\\
   & y_3=C_{33}x_3+C_{34}x_4
   \end{aligned}
\right. \\
& {\rm Agent\,4:}\left\{
   \begin{aligned}
   & \dot{x}_4=A_{44}x_4\\
   & y_4=C_{44}x_4
   \end{aligned}
\right. 
\end{align}
\end{subequations}
\end{footnotesize}
 % \vspace{-20pt} 

% \begin{footnotesize}
% \begin{align}\label{eq:example11}
% & {\rm Agent\,3:}\left\{
%    \begin{aligned}
%    & \dot{x}_3=A_{33}x_3+A_{34}x_4\\
%    & y_3=C_{33}x_3+C_{34}x_4
%    \end{aligned}
% \right. \tag{5c}\\
% & {\rm Agent\,4:}\left\{
%    \begin{aligned}
%    & \dot{x}_4=A_{44}x_4\\
%    & y_4=C_{44}x_4
%    \end{aligned}
% \right. \tag{5d}
% \end{align}
% \end{footnotesize}

\begin{figure}[!ht]
\centerline{\includegraphics[width=5.5cm]{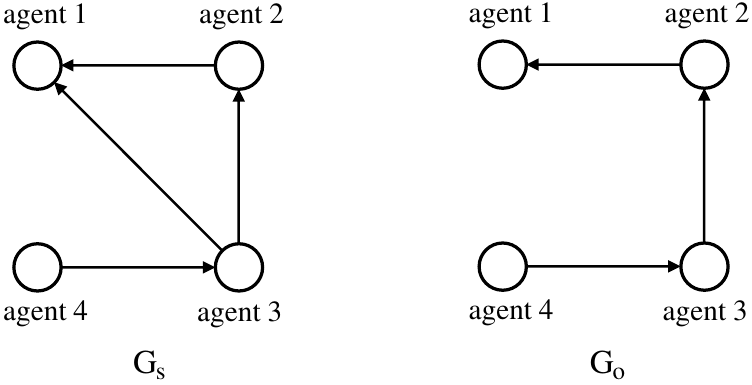}}

\caption{The graphs $\mathrm{G}_{\rm s}$ and $\mathrm{G}_{\rm o}$ of MAS \eqref{eq:example1}.}

\label{fig:example1}
\end{figure}
\end{example}

% \vspace{-10pt}
\section{Distributed State Estimation Scheme}\label{sec3}
Following the discussions in~\cite{yang2023plug} and \cite{chang2021resilient}, we consider the agents' control inputs to be bounded, which does not affect the stability of the estimation error dynamics of the distributed estimation algorithm. 
To present the core results, particularly the stability analysis of the estimation error dynamics, we first assume that the control input $u$ is known to all agents in the next three subsections.
This assumption is relaxed in Subsection~\ref{sub:discussion_to_unknown_input}, where we show that the proposed algorithm remains effective even when agents lack access to other agents’ inputs.

\subsection{Distributed Observer}\label{sec31}
The distributed observer consists of $m$ local observers, where each agent $i\in\mathcal{V}$ endows a local observer dynamics, with the following construction

\begin{footnotesize}
\begin{subequations}
\label{eq:Do}
\begin{align}
{\dot{\hat{x}}}_j^{(i)}=&{A}_{jj} {\hat{x}}_j^{(i)}+B_{jj}u_j+\sum_{l\in\mathcal N_j(\mathrm{G}_{\rm s})} {A}_{jl} {\hat{x}}_l^{(i)}\nonumber\\
     &+\mu\sum_{l\in \mathcal N_i(\mathrm{G}_{\rm c})}w_{il}^{(j)}\left({\hat{x}}_{j}^{(l)}- {\hat{x}}_j^{(i)}\right), {\rm ~for~}j\in \mathcal{V}/\{i\};
     \label{eq:O2}\\
{\dot{\hat{x}}}_i^{(i)}=&{A}_{ii} {\hat{x}}_i^{(i)}+B_{ii}u_i+\sum_{l\in\mathcal N_i(\mathrm{G}_{\rm s})} {A}_{il} {\hat{x}}_l^{(i)}+\mu\Big[w_{i0}^{(i)}\left({\bar{x}}_{i}- {\hat{x}}_i^{(i)}\right)\nonumber\\
  &+ \sum_{l\in \mathcal N_i(\mathrm{G}_{\rm c})} w_{il}^{(i)}\left({\hat{x}}_{i}^{(l)}- {\hat{x}}_i^{(i)}\right)\Big],
   \label{eq:O3}   
\end{align}
with
\begin{align}
{\dot{\bar{x}}}_i=& {A}_{ii} {\bar{x}}_i+B_{ii}u_i+\sum_{l\in\mathcal N_i(\mathrm{G}_{\rm s})} {A}_{il} {\hat{x}}_l^{(i)} \nonumber\\
 &+{F}_i \Big[{y}_i-({C}_{ii} {\bar{x}}_i+\sum_{l\in\mathcal N_i(\mathrm{G}_{\rm o})} {C}_{il} {\hat{x}}_l^{(i)})\Big].
 \label{eq:O1}
\end{align}
\end{subequations}
\end{footnotesize}In~\eqref{eq:Do},
$ {\hat{x}}_j^{(i)}\in \bb{R}^{n_j}$ denotes the estimate to $x_j$ held by node $i$ and
${\bar{x}}_i\in \bb{R}^{n_i}$ is an auxiliary variable interpreted as an estimate of $x_i$ by agent $i$ using a Luenberger-like estimation scheme; $\mathcal N_i(\mathrm{G}_{\rm c})$ denotes agent $i$'s in-neighbor set in graph $\mathrm{G}_{\rm c}$;
${F}_i\in \bb{R}^{n_i\times p_i}$ is the Luenberger gain, $\mu$ is the coupling gain; $w_{i0}^{(i)}$ and $w_{il}^{(j)}$ for $i,j\in\mathcal V, l\in\mathcal N_i(\mathrm{G}_{\rm c})$ are nonnegative consensus gains related to $\mathrm G_{\rm c}$.
Their design will be discussed in more detail in Section~\ref{sec32}.
To enable~\eqref{eq:Do}, each agent $i$ must transmit its estimate $\hat x^{(i)}$ of the entire MAS state to its out-neighbors in communication graph $\mathrm G_{\rm c}$, where $\hat x^{(i)}=[({\hat{x}}_1^{(i)})^\top, \ldots,  ({\hat{x}}_m^{(i)})^\top]^\top$ comprises agent $i$'s estimate $\hat{x}_i^{(i)}$ of its own state, obtained from \eqref{eq:O3}, and its estimate $\hat{x}_j^{(i)}$ of the states of other agents, obtained from \eqref{eq:O2} for $\forall j\in\mathcal{V}/\{i\}$.
It can be seen from the structure of the proposed distributed observer~\eqref{eq:Do} that each agent $i$ can estimate all agents' states in the MAS only based on the local measurement $y_i$ and the exchanged information $\hat x^{(j)}$ from in-neighbors $j\in\mathcal N_i(\mathrm{G}_{\rm c})$.

The structure of the proposed distributed observer~\eqref{eq:Do} is depicted in Fig.~\ref{fig:design}. 
In particular, the local observer embedded in each agent $i\in\mathcal V$ consists of two parts: the consensus-based estimation rules \eqref{eq:O2}-\eqref{eq:O3}, which involve distributed interaction and are used to compute the estimate $\hat x^{(i)}$, where $\hat x^{(i)}$ is the only information agent $i$ transmit to its communication neighbors; and the Luenberger-like estimation rule \eqref{eq:O1}, which computes $\bar{x}_i$ and interacts only with $\hat x^{(i)}$.
Consequently, $\bar{x}_i$ is invisible to the neighbors of agent $i$, although it directly affects the evolution of $\hat x^{(i)}$.
In the process of all agents estimating $x_i$ using \eqref{eq:O2}-\eqref{eq:O3}, the extended state $\bar{x}_i$ can be regarded as a leader signal tracked by all agents to obtain the estimates ${\hat{x}}_i^{(1)}, {\hat{x}}_i^{(2)},\ldots, {\hat{x}}_i^{(m)}$.
This design is motivated by the leader-follower consensus technique \cite{cao2015leader}.
However, it should be noted that the stability analysis for~\eqref{eq:Do} is more involved than the leader-follower consensus as 
\eqref{eq:O2}-\eqref{eq:O3} and \eqref{eq:O1} are dynamically coupled.

\begin{figure}[!ht]
\centerline{
\includegraphics[width=1.1in]{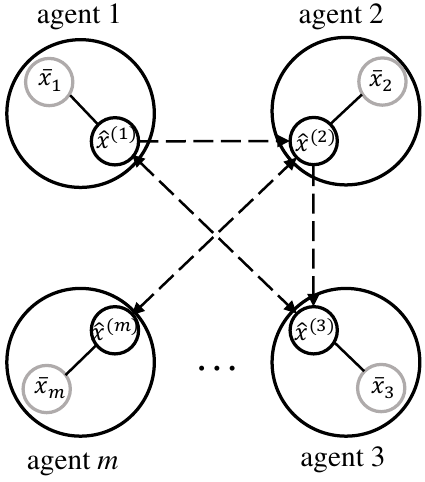}}
\caption{The structure of the distributed observer \eqref{eq:Do}.}
\label{fig:design}
\end{figure}

\subsection{Stabilizability Analysis}\label{sec32}
In this part, we explore conditions for the estimation error dynamics of the distributed observer \eqref{eq:Do} to be stabilized.
To this end, we first introduce a new augmented graph for each agent $j\in \mathcal{V}$. 
Define the augmented graph as $\mathrm{G}_{\rm c}^{(j)}=(\overline{\mathcal{V}}, \mathcal{E}_{\rm c}^{(j)})$, where
the node set is $\overline{\mathcal{V}}=\mathcal V \cup \{0\}$ and the edge set $\mathcal{E}_{\rm c}^{(j)} = \mathcal E_{\rm c}\cup \{(0,j)\} \subseteq\overline{\mathcal{V}}\times\overline{\mathcal{V}}$.
When we associate node $0$ with the dynamics of $ {\bar{x}}_j$ and associate
node $i$, where $i\in\mathcal V$,
 with the dynamics of $ {\hat{x}}_j^{(i)}$, the augmented graph $\mathrm{G}_{\rm c}^{(j)}$ reflects the way that nodes $0$ to $m$ are coupled in the dynamics of 
 ${\bar{x}}_j$ and
${\hat{x}}_j^{(1)},\ldots, {\hat{x}}_j^{(m)}$, described by~\eqref{eq:O1} and~\eqref{eq:O2}-\eqref{eq:O3}.
An example of $\mathrm{G}_{\rm c}$ and $\mathrm{G}_{\rm c}^{(j)}$ (where $j=2$) is given in Fig.~\ref{fig:G_G}.
Notice that $\mathrm{G}_{\rm c}$ is a subgraph of the graph $\mathrm{G}_{\rm c}^{(j)}$ and can be obtained by removing node $0$ from $\overline{\mathcal{V}}$ and edge $(0,j)$ from $\mathcal{E}_{\rm c}^{(j)}$.
We align a weight matrix  $W^{(j)}$ to $\mathrm{G}_{\rm c}^{(j)}$ so that it interacts with $\mathcal{E
}_{\rm c}^{(j)}$ in the following way: $(r, s) \in \mathcal{E}_{\rm c}^{(j)}$ if and only if $w_{sr}^{(j)}>0$.

%Considering the estimation of the state of agent $j\in\mathcal V$ by all agent, based on the proposed observer \eqref{eq:Do}, we can define a nonnegative matrix $ {W}^{(j)}=[w_{rs}^{(j)}]\in \bb{R}^{(m+1)\times (m+1)}$, $r,s=0,1,\ldots,m$,

\begin{figure}[!ht]
\centerline{\includegraphics[width=2.1in]{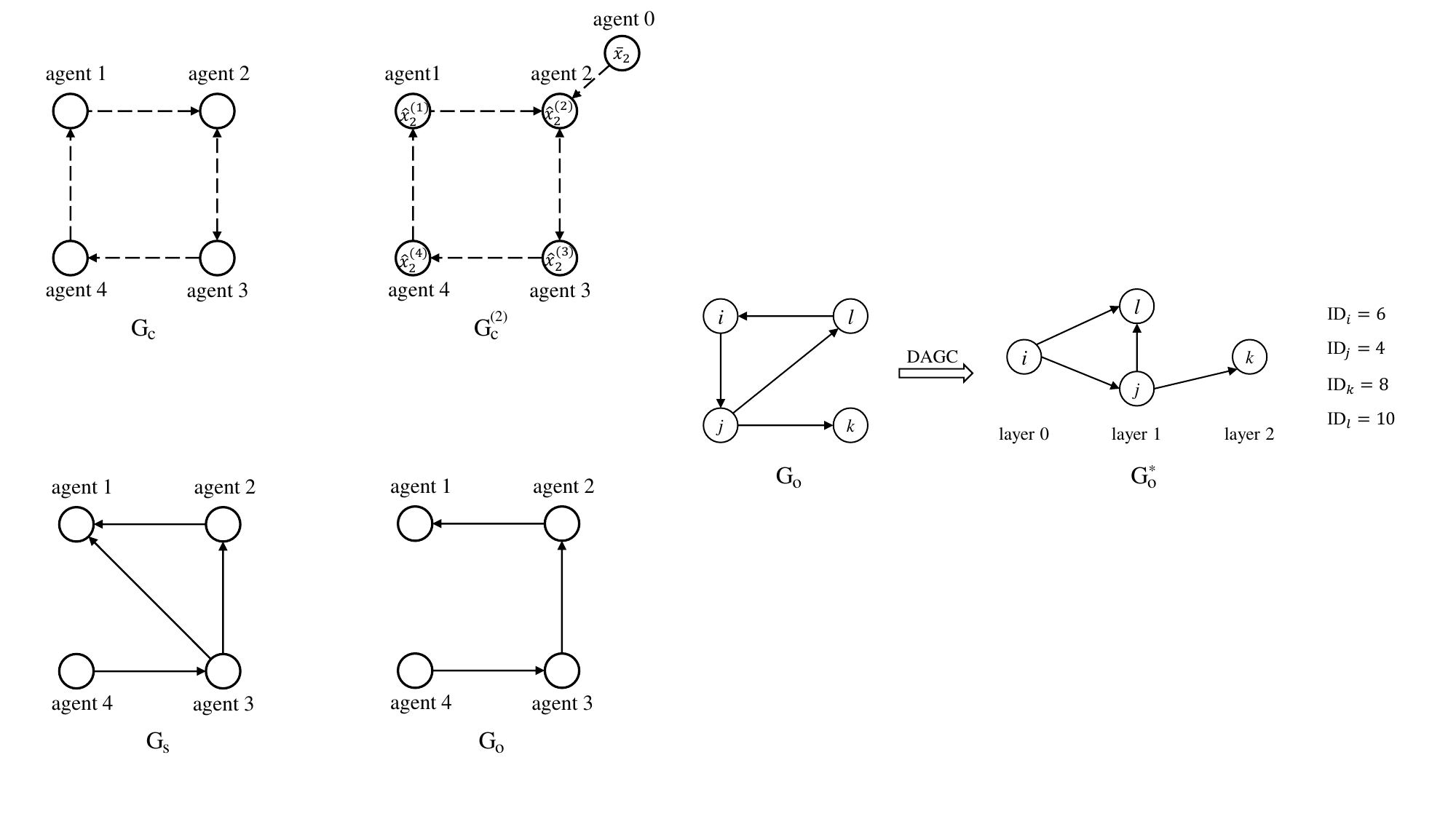}}
\caption{An example of a communication graph $\mathrm{G}_{\rm c}$ with four agents and the augmented graph $\mathrm{G}_{\rm c}^{(j)}$ ($j=2$).}
\label{fig:G_G}
\end{figure}

We construct the Laplacian matrix from $W^{(j)}$ as $ L_{\mathrm{G}_{\rm c}^{(j)}}=\diag\left(\sum\limits_{l=0}^{m}w_{0l}^{(j)},\ldots,\sum\limits_{l=0}^{m}w_{ml}^{(j)}\right)-W^{(j)}$. In accordance with the partitioning of
$\overline{\mathcal{V}}$ into
$\{0\}$ and $\mathcal V$, it can also be 
partitioned into four blocks as follows:
\begin{small}
\begin{equation}
\label{eq:definition L}
  L_{\mathrm{G}_{\rm c}^{(j)}}=\left[\begin{array}{c|c}
 0 &  \mathbf{0}\\ 
 \hline
 -{O}^{(j)}& S^{(j)}
\end{array}\right], 
\end{equation}
\end{small}where $O^{(j)}=w_{j0}^{(j)}b_j$. 
The first row of $L_{\mathrm{G}_{\rm c}^{(j)}}$ is a $1\times (m+1)$ all-zeros row vector, which is consistent with the in-degree of node 0 being zero in graph $\mathrm{G}_{\rm c}^{(j)}$.
Notice that $S^{(j)}\in \bb{R}^{m\times m}$ is aligned to the subgraph $\mathrm G_{\rm c}$.
Since $L_{\mathrm{G}_{\rm c}^{(j)}}\mathbf{1}_{m+1}= 0$, we have $S^{(j)}\mathbf{1}_m=O^{(j)}$.
We have the following findings about 
${S}^{(j)}$.
\begin{lemma}
\label{le:S(j)}
For any given $j\in\mathcal V$, 
all the eigenvalues of $ {S}^{(j)}$ have strictly-positive real parts if and only if  $\mathrm{G}_{\rm c}$ is strongly connected.
\end{lemma}
\begin{proof}
According to~\eqref{eq:O1}, we have 
$w_{0i}^{(j)}$'s are all zeros for all $i\in\mathcal V$. Therefore, $ \sum_{i=1}^{m} w_{0 i}^{(j)}=0$ and the nonzero eigenvalues of $ {S}^{(j)}$ coincide with those of $ {L}_{\mathrm{G}_{\rm c}^{(j)}}$.
By Lemma~\ref{le:spectrum of L}, all the nonzero eigenvalues of $ {L}_{\mathrm{G}_{\rm c}^{(j)}}$ have positive real parts. So do the nonzero eigenvalues of $ {S}^{(j)}$.
 By Lemma~\ref{le:spectrum of L} again, $ {S}^{(j)}$ is nonsingular if and only if $\mathrm{G}_{\rm c}^{(j)}$ has a spanning tree. Since $ {S}^{(j)}$ should be nonsingular for any given $j\in\mathcal V$, it is equivalent to that $\mathrm{G}_{\rm c}$ is strongly connected, which completes the proof.
\end{proof}

Now we give the result on stabilizability of the estimation error dynamics of our proposed observer \eqref{eq:Do} and the 
connectivity of the underlying communication graph. 
\begin{theorem}
\label{theorem:main1}
Consider a multi-agent system \eqref{eq:sys1}-\eqref{eq:sys2} satisfying Assumptions \ref{ass:local observability}-\ref{ass:Gs-Go} and suppose that $A_{ii}$ is not Hurwitz for all $i\in\mathcal{V}$.
Let $ {F}_i$ be such that $ {A}_{ii}- {F}_i {C}_{ii}$ is Hurwitz for all $i\in\mathcal{V}$.
Then the distributed observer \eqref{eq:Do} with sufficiently large positive number $\mu$ is stable in the sense that $\lim_{t\to \infty} \|\hat{x}^{(i)}-x\|=0$ for all $i\in \mathcal{V}$, if and only if the communication graph $\mathrm{G}_{\rm c}$ is strongly connected.
\end{theorem}

\begin{proof}
Let $ {\bar{e}}_{j} \triangleq  {\bar{x}}_{j}- {x}_j$ and $ {e}_j^{(i)} \triangleq  {\hat{x}}_j^{(i)}- {x}_j$, where $i,j\in \mathcal{V}$.
By \eqref{def:MAS_dynamics}, \eqref{def:MAS_observation} and \eqref{eq:Do}, the error dynamics of $ {\bar{e}}_{j}$ and $ {e}_j^{(i)}$ can be written as
\begin{footnotesize}
\begin{subequations}
\label{eq:error dynamic}
\begin{align}
{\dot{e}}_{j}^{(i)}=&{A}_{jj} {e}_{j}^{(i)}+\sum_{l\in\mathcal N_j({\mathrm{G}_{\rm s}})} {A}_{jl} {e}_l^{(i)}\nonumber\\
&+\mu\sum_{l\in \mathcal N_i({\mathrm{G}_{\rm c}})}w_{il}^{(j)}\left( {e}_{j}^{(l)}- {e}_{j}^{(i)}\right),\hbox{~for~} i \neq j\in \mathcal{V};
\label{eq:e_i^j}\\
{\dot{e}}_{j}^{(j)}=&{A}_{jj} {e}_{j}^{(j)}+\sum_{l\in\mathcal N_j({\mathrm{G}_{\rm s}})} {A}_{jl} {e}_l^{(j)}\nonumber\\
&+\mu\Big[w_{j0}^{(j)}\left( {\bar{e}}_j- {e}_j^{(j)}\right )+\sum_{l\in \mathcal N_j({\mathrm{G}_{\rm c}})}w_{jl}^{(j)}\left( {e}_{j}^{(l)}- {e}_{j}^{(j)}\right)\Big],\label{eq:e_j^j}
\end{align}
with
\begin{align}
{\dot{\bar{e}}}_{j}=&\left({A}_{jj}-{F}_j {C}_{jj}\right){\bar{e}}_{j}+\sum_{l\in \mathcal N_j({\mathrm{G}_{\rm s}})} {A}_{jl} {e}_l^{(j)}-\sum_{l\in \mathcal N_j({\mathrm{G}_{\rm o}})} {F}_j {C}_{jl} {e}_l^{(j)}.
\label{eq:e_0^j}
\end{align}
\end{subequations}
\end{footnotesize}

Let $E_j\triangleq[(\bar{e}_j)^\top, (e_{j}^{(1)})^\top \ldots, (e_{j}^{(m)})^\top]^\top$ be the overall estimation error of all agents estimating the state of agent $j$.
Based on \eqref{eq:definition L} and \eqref{eq:error dynamic}, we obtain the following dynamics of $E_j$.
\begin{small}
\begin{equation}
\label{eq:e^j1}
\begin{aligned}
 {\dot{E}}_j= {T}^{(j)} E_j+\sum_{l\in \mathcal N_j({\mathrm{G}_{\rm s}})\cup \mathcal N_j({\mathrm{G}_{\rm o}})} {Q}^{(jl)}E_l
\end{aligned}
\end{equation}
\end{small}where 
{\footnotesize{${T}^{(j)}=\left[\begin{array}{c|c}
 {A}_{jj}- {F}_j {C}_{jj} &  \mathbf{0}\\
\hline
- {O}^{(j)} \otimes {I}_{n_j} &  {I}_m\otimes {A}_{jj}-\mu {S}^{(j)}\otimes {I}_{n_j}\end{array}\right]$}} is associated with the communication graph $\mathrm{G}_{\rm c}$ due to the presence of matrices $- {O}^{(j)}$ and ${S}^{(j)}$, while
\begin{footnotesize}
\begin{equation}
\label{eq:qjl}
\begin{aligned}\nonumber
 {Q}^{(jl)} &=\left\{\begin{array}{lllllll}
\left[\begin{array}{c|c}
 \mathbf{0}_{n_j\times n_l} &  {b}_l^\top \otimes ( {A}_{jl}- {F}_j {C}_{jl})\\
\hline
 \mathbf{0}_{mn_j\times n_l} & {I}_m\otimes  {A}_{jl}
\end{array}\right],  \\
\qquad\qquad\qquad\qquad\hbox{~for ~} l\in \mathcal N_j({\mathrm{G}_{\rm s}})\cap\mathcal N_j({\mathrm{G}_{\rm o}}),\\
\\
\left[\begin{array}{c|c}
\mathbf{0}_{n_j\times n_l} &  \mathbf{0}_{n_j\times mn_l}\\
\hline
\mathbf{0}_{mn_j\times n_l} & {I}_m\otimes  {A}_{jl}
\end{array}\right],\\
\qquad\qquad\qquad\qquad\hbox{~for~} l\in \mathcal N_j({\mathrm{G}_{\rm s}})/\mathcal N_j({\mathrm{G}_{\rm o}}),\\
\\
\left[\begin{array}{c|c}
 \mathbf{0}_{n_j\times n_l}  &  {b}_l^\top\otimes(- {F}_j {C}_{jl})\\
\hline
\mathbf{0}_{mn_j\times n_l}  &  \mathbf{0}_{mn_j\times mn_l} 
\end{array}\right], \\
\qquad\qquad\qquad\qquad\hbox{~for~} l\in \mathcal N_j({\mathrm{G}_{\rm o}})/\mathcal N_j({\mathrm{G}_{\rm s}}),
\end{array}
\right.
\end{aligned}
\end{equation}
\end{footnotesize}
and the symbol $\otimes$ denotes a Kronecker product.

Then, we define the whole estimation error as ${E}\triangleq[({E}_{1})^\top, \ldots, ({E}_{m})^\top]^\top$. 
By Assumption \ref{ass:Gs-Go}, we derive a topological ordering of $\mathrm{G}_{\rm s}$, denoted as $\{s_1, s_2, \ldots, s_m\}$, which is also the topological ordering of $\mathrm{G}_{\rm o}$. 
According to this ordering, we permute the rows of the whole estimation error $E$ and obtain the following dynamics

\begin{footnotesize}
\begin{equation}
\label{eq:entire error}
\begin{aligned}
 {\dot{\tilde{E}}}=\underbrace{
\left[\begin{array}{cccccc}
 {T}^{(s_1)} & & & \\
 {Q}^{(s_2s_1)}&  {T}^{(s_2)} && \\
\vdots&\ddots & \ddots&\\
 {Q}^{(s_ms_1)}&\cdots& {Q}^{(s_ms_{m-1})}&  {T}^{(s_m)}
\end{array}\right]}_{ {R}}
 {\tilde{E}},
\end{aligned}
\end{equation}
\end{footnotesize}where $\tilde{E}=[(E_{s_1})^\top, \ldots, (E_{s_m})^\top]^\top$.

{\bf If.}~Considering the lower triangular structure of the system matrices $R$ in \eqref{eq:entire error} and $T^{(j)}$ in \eqref{eq:e^j1}, we can make the whole estimation error dynamics \eqref{eq:entire error} stable by stabilizing the matrices $ {A}_{jj}- {F}_j {C}_{jj}$ and $ {I}_m\otimes {A}_{jj}-\mu {S}^{(j)}\otimes {I}_{n_j}$ for all $j\in \mathcal V$.
Under Assumption \ref{ass:local observability}, there always exists an appropriate gain $ {F}_j$ such that $ {A}_{jj}- {F}_j {C}_{jj}$ is Hurwitz for any $j\in\mathcal V$. 
On the other hand, for $\forall j\in\mathcal V$, the eigenvalues of $ {I}_m\otimes {A}_{jj}-\mu {S}^{(j)}\otimes {I}_{n_j}$ are 
$\{\lambda_p( {A}_{jj})-\mu\lambda_q( {S}^{(j)})\}_{p\in\{1, \ldots, n_j\}, q\in\mathcal V}$
where $\lambda_p( {A}_{jj})$ and $\lambda_q( {S}^{(j)})$ are the eigenvalues of $ {A}_{jj}$ and $ {S}^{(j)}$, respectively. 
By Lemma \ref{le:S(j)}, the real parts of the eigenvalues of ${S}^{(j)}$ are all positive when the communication graph $\mathrm G_{\rm c}$ is strongly connected.
Thus we can choose $\mu$ sufficiently large so that $Re(\lambda_p( {A}_{jj})-\mu\lambda_q( {S}^{(j)}))<0, \forall p,q$ makes $ {T}^{(j)}$ Hurwitz for all $j\in\mathcal V$, and hence, the matrix $ {R}$ is Hurwitz, i.e., $\lim _{t \rightarrow \infty}E(t)=0$.

{\bf Only if.}~If the communication graph $\mathrm{G}_{\rm c}$ is not strongly connected, then there exist agents $j, k \in \mathcal{V}$ such that no directed path exists from $j$ to $k$.
This implies that the graph $\mathrm{G}_{\rm c}^{(j)}$ contains no spanning tree rooted at agent $j$. 
By Lemma~\ref{le:spectrum of L}, the Laplacian matrix $L_{\mathrm{G}_{\rm c}^{(j)}}$ has at least two zero eigenvalues, implying that the matrix ${S}^{(j)}$ has at least one zero eigenvalue.
Since $A_{jj}$ is not Hurwitz by assumption (the Hurwitz case is trivial), the matrix ${I}_m\otimes {A}_{jj}-\mu {S}^{(j)}\otimes {I}_{n_j}$ has at least one eigenvalue with a non-negative real part. 
Thus, ${T}^{(j)}$ is not Hurwitz, which implies that the matrix $R$ is also not Hurwitz, thereby completing the proof.\end{proof}

It is worth noting that each agent $i$ can determine the Luenberger gain $F_i$ and the consensus gains $w_{i0}^{(i)}$ and $w_{il}^{(j)}$ for $l\in \mathcal N_i(\mathrm{G}_{\rm c})$ in \eqref{eq:Do} by itself.
Specifically, each agent $i$ can choose $F_i$ freely as long as $ {A}_{ii}- {F}_i {C}_{ii}$ is Hurwitz, and set the consensus gains $w_{i0}^{(i)}$ and $w_{il}^{(j)}$ according to
the $i$th row of weight matrices $W^{(j)}$'s aligned to ${\rm G}_{\rm c}^{(j)}$. 
The weight matrix $W^{(j)}$ can be constructed 
in different ways. For example, it can be constructed as a binary 
weight matrix as constructed in \cite{su2011cooperative} with 
$w_{sr}^{(j)}=1$ if and only if $(r, s) \in \mathcal{E}_{\rm c}^{(j)}$;
or it can be constructed as a normalized weight matrix 
as constructed in~\cite{pu2020push}:

\begin{footnotesize}
\begin{equation}
\label{eq:consensusgain1}
w_{sr}^{(j)}=\left\{\begin{array}{cll}
\frac{1}{|
{\mathcal N}_s({\mathrm G}_{\rm c}^{(j)})|}, &\hbox{~for~} (r,s)\in{\mathcal E}_{\rm c}^{(j)},\\
0, &\hbox{~for~} 
(r,s)\not\in{\mathcal E}_{\rm c}^{(j)} \hbox{~or~} r=s.
\end{array}\right.
\end{equation}
or 
\begin{equation}
\label{eq:consensusgain2}
w_{sr}^{(j)}=\left\{\begin{array}{cll}
\frac{1}{|
{\mathcal S}_r({\mathrm G}_{\rm c}^{(j)})|}, &\hbox{~for~} (r,s)\in{\mathcal E}_{\rm c}^{(j)},\\
0, &\hbox{~for~} 
(r,s)\not\in{\mathcal E}_{\rm c}^{(j)} \hbox{~or~} r=s.
\end{array}\right.
\end{equation}
\end{footnotesize}The selection of the coupling gain $\mu$ depends on the eigenvalues of the matrices $S^{(j)}$'s, which is global information related to the topology of ${\mathrm G}_{\rm c}$. 
This hinders decentralized construction of our distributed observer \eqref{eq:Do} as well as plug-and-plug operation resilient to changes of MASs.
In the next part, we will give thought to possible ways to cope with this issue.

\subsection{Fully Distributed Design}\label{sec33}
We discuss how to choose the coupling gain $\mu$ locally for each agent and present a fully distributed design of the proposed observer \eqref{eq:Do} in some cases.
From the proof of Theorem~\ref{theorem:main1}, it can be seen that to stabilize the estimation error of \eqref{eq:Do}, the coupling gain $\mu$ should satisfy 

\begin{footnotesize}
\begin{equation}
\label{eq:mu}
\mu>\frac{\max\limits_{j\in\mathcal V}\rho(A_{jj})}{\min \limits_{q,j\in \mathcal V} |\lambda_q( {S}^{(j)})|},
\end{equation}
\end{footnotesize}where $\rho(A_{jj})$ is the spectral radius of $A_{jj}$. 

Note that $\max\limits_{j\in\mathcal V}\rho(A_{jj})$ can be known to each agent from their knowledge of the system matrix $A$ of the MAS (for instance, in a
homogeneous MAS we always have $\max_{j\in\mathcal V}\rho(A_{jj})=\rho(A_*)$, where 
$A_*$ is the homogeneous individual system matrix), however, how to estimate the minimal modulus of the eigenvalues of ${S}^{(j)}$ locally needs more effort.
We next provide solutions to this problem in the case of undirected communication graphs and directed communication graphs, and on this basis, we present local methods for selecting coupling gain $\mu$ by each agent in the two cases.

\subsubsection{The case of undirected communication graph} 
The following lemma provides a lower bound for the minimal modulus of the eigenvalues of ${S}^{(j)}$ in case of an undirected $\mathrm{G}_{\rm c}$ and binary $W^{(j)}$.

\begin{lemma}
\label{le:mu}
Suppose that the communication graph $\mathrm{G}_{\rm c}$ is an undirected connected graph 
and $W^{(j)}$ is the binary weight matrix
associated with the augmented graph ${\rm G}_{\rm c}^{(j)}$.
Denote the  eigenvalues of Laplacian matrix $L_{\mathrm{G}_{\rm c}}$ as $\lambda_1, \ldots, \lambda_m$ sorted in ascending order.
Then, 

\begin{footnotesize}
\begin{equation}
\label{eq:eigv of Sj}
\min \limits_{q,j\in \mathcal V} \lambda_q( {S}^{(j)})\geq\frac{1}{m}(\frac{\lambda_2}{\lambda_2+1})^{m-1}.
\end{equation}
\end{footnotesize}
\end{lemma}

\begin{proof}
Let $\overline{\mathrm{G}}_{\rm c}^{(j)}$ be the undirected graph obtained by dropping the direction of edge $(0,j)$ in $\mathrm{G}_{\rm c}^{(j)}$, as $\mathrm{G}_{\rm c}$ is assumed to be undirected.
With the decomposition of
the Laplacian $L_{{\rm G}_{\rm c}^{(j)}}$ in~\eqref{eq:definition L},
we can write the Laplacian for $\overline{\mathrm{G}}_{\rm c}^{(j)}$ 
as 
$L_{\overline{\mathrm{G}}_{\rm c}^{(j)}}=\left[\begin{array}{c|c}
 1 &  -b_j^\top\\
\hline
-b_j & S^{(j)}
\end{array}\right].$ 
It is clear that the numbers of the spanning trees in graphs $\mathrm{G}_{\rm c}$ and $\overline{\mathrm{G}}_{\rm c}^{(j)}$ are the same for any $j\in\mathcal V$.
Without loss of generality, 
we let the eigenvalues of $S^{(j)}$ be sorted in ascending order.
By Kirchhoff's matrix tree theorem \cite{biggs1993algebraic}, 
the former number equals to
$\frac{1}{m}\prod_{q=2}^{m}\lambda_q$, while the latter one 
$\det(S^{(j)})=\prod_{q=1}^{m} \lambda_q( {S}^{(j)})$ as 
the determinant of the submatrix constructed by deleting the first row and the first column of $L_{\overline{\mathrm{G}}_{\rm c}^{(j)}}$. Then,
% \begin{small}
% \begin{equation}
% \nonumber
$\lambda_1( {S}^{(j)})=\frac{\prod_{q=2}^{m}\lambda_q}{m\prod_{q=2}^{m} \lambda_q( {S}^{(j)})}.$
% \end{equation}
% \end{small}
On the other hand, by Weyl's inequality \cite{horn2012matrix}, we have 
% \begin{align}
$\lambda_q\le\lambda_q( {S}^{(j)})\le\lambda_q+1,~~\forall j,q\in\mathcal V,$
% \end{align}
which combining the above equality further yields that 
% \begin{small}
% \begin{equation}
% \nonumber
$\lambda_1( {S}^{(j)})
\ge\frac{1}{m}\prod_{q=2}^{m}\frac{\lambda_q}{\lambda_q+1}
\ge\frac{1}{m}(\frac{\lambda_2}{\lambda_2+1})^{m-1}.$
% \end{equation}
% \end{small}
Therefore, 
$
\min \limits_{q,j\in \mathcal V} \lambda_q( {S}^{(j)})\ge\frac{1}{m}(\frac{\lambda_2}{\lambda_2+1})^{m-1}$, which completes the proof.
\end{proof}

In virtue of Lemma~\ref{le:mu}, for a MAS with a maximum number $\overline{m}$ of allowable agents, if each agent is aware of $\overline{m}$,
beforehand,
they can uniformly take $\mu$ as the minimum integer satisfying 
\begin{align}
\label{eq:minmu}
\mu>\max\limits_{j\in\mathcal V}\rho(A_{jj})
\left(
\frac{\overline{m}^2-\overline m+4}{4}
\right)^{\overline{m}-1}\overline{m}.
\end{align}
This result is based on \eqref{eq:eigv of Sj} and the fact $\lambda_2\ge4/(m^2-m)$ from~\cite{mohar1991eigenvalues}.
Using the above parameter selection method, each agent can automatically reconfigure to stabilize the estimation error dynamics when facing the node changes of the MAS.
Since the estimation error dynamics can be stabilized both before and after system changes and these changes occur sporadically, employing the stability results from \cite{hespanha1999stability,bai2009instability} for switched systems, it follows that the error dynamics of~\eqref{eq:Do} are always stable as long as the proposed assumptions are valid during the operation, though some agents may sporadically join or leave the MAS.

\subsubsection{The case of directed communication graph} 
For the case where $\mathrm{G}_{\rm c}$ is a directed graph with a normalized $W^{(j)}$, we derive a lower bound for $\mu$.
\begin{lemma}
\label{le:ccc2012huang}
Suppose that the communication graph $\mathrm G_{\rm c}$ is a directed and strongly connected graph with a normalized weight matrix.
Then, 
\begin{equation}
\label{eq:eigv of Sj1}
\min \limits_{q,j\in \mathcal V} |\lambda_q( {S}^{(j)})|\geq 1-(1-\frac{1}{(m+1)!})^{\frac{1}{m}}.
\end{equation}
\end{lemma}
\begin{proof}
The result follows directly from Corollary 1 of~\cite{chao2012lower}.
\end{proof}
For a MAS with a maximum of $\overline{m}$ allowable agents, the agents can choose a uniform $\mu$ as the smallest integer satisfying

\begin{small}
\begin{equation}
\label{eq:minmu1}
\mu>
\max\limits_{j\in\mathcal V}\rho(A_{jj})
\left(1-(1-\frac{1}{(\overline{m}+1)!})^{\frac{1}{\overline{m}}}\right)^{-1},
\end{equation}
\end{small}which guarantees a fully distributed design of the proposed observer \eqref{eq:Do}.

% {\color{blue}
% \begin{remark}
% \label{re:other pnp algorithms}
% The problem formulation in this paper differs from those in \cite{kim2019completely}, \cite{yang2023plug} and \cite{siljak2013decentralized}.
% Unlike in \cite{kim2019completely} and \cite{yang2023plug}, where the dynamics of the observed LTI system are independent of the monitoring agents, the observed MAS in this paper consists of all monitoring agents, with ``state-to-state'' and ``state-to-output'' couplings among them.
% Moreover, unlike in \cite{siljak2013decentralized}, where the decentralized observer only estimates the state of its respective subsystem, our approach estimates the entire state of the MAS.
% \end{remark}
% }

\subsection{Distributed Estimation with Inaccessible Inputs}\label{sub:discussion_to_unknown_input}
In general, the inputs of other agents, $u_j$ for $j\neq i$, are not available to agent $i$.
Nevertheless, the proposed distributed estimation framework remains effective. 
Specifically, in the observer \eqref{eq:Do}, equations \eqref{eq:O3} and \eqref{eq:O1} remain unchanged, since agent $i$ has access to its own input $u_i$.
However, equation \eqref{eq:O2}, which governs agent $i$'s estimation of other agents' states, needs to be modified as follows:

\begin{footnotesize}
\begin{align}
{\dot{\hat{x}}}_j^{(i)}=&{A}_{jj} {\hat{x}}_j^{(i)}+\sum_{l\in\mathcal N_j(\mathrm{G}_{\rm s})} {A}_{jl} {\hat{x}}_l^{(i)}+\mu\sum_{l\in \mathcal N_i(\mathrm{G}_{\rm c})}w_{il}^{(j)}\left({\hat{x}}_{j}^{(l)}- {\hat{x}}_j^{(i)}\right).
\end{align}
\end{footnotesize}Given that the inputs $u_j$ are bounded, i.e., $\left\| u_j\right\|\leq\bar{u}$ for all $j\in\mathcal{V}$, a stability analysis based on the input-to-state stability (ISS) framework--similar to that in Subsection \ref{sec32}--shows that the estimation error are bounded.
In particular, the estimation error satisfies 
\begin{align}\label{eq:bounded_estimation_error}
\left\| E(t)\right\| \leq \kappa \exp{(-\eta t)}\left\| E(0)\right\|+\frac{\kappa\left\| B\right\|}{\eta}\bar{u},
\end{align}
where $\kappa\geq 1$ and $\eta>0$ are constants determined by the Hurwitz property of the system matrix $R$ in \eqref{eq:entire error}, such that $\left\|\exp(Rt)\right\|\leq\kappa\exp(-\eta t)$ for all $t\geq 0$.
Therefore, under the same conditions as Theorem~\ref{theorem:main1}, the estimation error is bounded in the case of partial input information.
Similar results hold when the proposed distributed estimation algorithm is applied to a MAS with bounded process and measurement noise.
For brevity, we omit the detailed analysis, but these findings are validated through simulations (see Figs.~\ref{fig:heteroMASnoise} and \ref{fig:singleunknown}).

\section{Application in Cooperative Localization Problems}\label{sec4}

In this section, we discuss the application of our distributed estimator to cooperative localization problems. In a cooperative localization problem, only a small proportion of the agents are able to know their absolute positions, while the rest have access to relative positions with respect to
ambient nodes; the whole group of nodes will exchange information locally to find their positions jointly . 
In our problem setup, the nodes in question are mobile. We model an agent's kinematics using the dynamics of a single integrator and the whole systems using a MAS~\eqref{def:MAS_dynamics}-\eqref{def:MAS_observation}.
We first analyze the observability of the MAS and then propose a fully distributed localization algorithm.
We also establish necessary and sufficient conditions to make the localization process stable as the nodes move. 
To enable the algorithm, we need to distributedly construct an acyclic sensing graph with the aid of the local communication network. 
We finally extend our algorithm to double-integrator modeled agents as we may also be interested in the mobile nodes' velocities.

\subsection{Localization Problem for Single-integrator modeled Agents}\label{sec41}
Consider a group of $m$ single-integrator modeled agents moving in a 2D or 3D space, where the communication graph is described as $\mathrm G_{\rm c}$ and agent $i$ has the following dynamics:
\begin{equation}
\label{eq:single integrator}
 {\dot{p}}_i= {u}_i,\quad i\in\mathcal V
\end{equation}
where $p_i\in\bb{R}^{h}$ and $u_i\in\bb{R}^{h}$ denote the position and the control input of agent $i$ with respect to the global reference frame denoted by $^g\sum$, and $h=2$ or $3$.

Each agent $i\in\mathcal V$ has its local reference frame with the origin at $ {p}_i$, which is denoted by $^i\sum$.
In our problem setup, agents can sense the relative positions with respect to their neighbors in $\mathrm G_{\rm o}$ in reference frame $^i\sum$, but only a proportion of them know the absolute positions in $^g\sum$.
We consider that the orientation of $^i\sum$ is aligned to that of $^g\sum$ for all $i\in\mathcal V$, which can be realized using magnetic sensors in practice~\cite{oh2015survey}.
The relative positions with respect to local reference frames can be converted into those with respect to the global reference frame due to the alignment of the reference frames.
The sensing topology is described as $\mathrm G_{\rm o}$.
Hence, we model the relative position between agents $i$ and $j$ measured by agent $i$ as $ {p}_{ij}= {p}_j- {p}_i$ for $i\in \mathcal V$, $j\in \mathcal N_i(\mathrm G_{\rm o})$, see
Fig.~\ref{fig:setup} as illustration. In compact, the measurement of agent $i$ can be written as
\begin{equation}
\label{eq:single-yi}
 {y}_i= {C}_{ii} {p}_i+\sum_{j\in\mathcal N_i(\mathrm G_{\rm o})} {C}_{ij} {p}_j,
\end{equation}
where ${y}_i\in\bb{R}^{q_i}$, ${C}_{ii}\in\bb{R}^{q_i\times h}$ and $ {C}_{ij}\in\bb{R}^{q_i\times h}$ are the output, observation matrix and output coupling matrix of agent $i$, respectively. 
Note that the matrices $C_{ii}$ and $C_{ij}$ in \eqref{eq:single-yi} are only composed of the matrices $I_h$, $-I_h$ and $\mathbf 0_{h\times h}$.
The dynamics of the single-integrator multi-agent system can be written as
\begin{align}
&\dot{p}=u \label{eq:m-agent}\\
&y=Cp \label{eq:y-m-agent}
\end{align}
where $p=[p_1^\top,\ldots, p_m^\top]^\top$, $u=[u_1^\top,\ldots, u_m^\top]^\top$, $y=[y_1^\top,\ldots, y_m^\top]^\top$ and $C=\left[ C_{ij}\right]_{i,j\in\mathcal V}$ are the state, input, output and observation matrix of the single-integrator MAS, respectively.
The cooperative localization problem discussed here is to design a fully distributed localization algorithm such that each agent $i$ in set $\mathcal V$ is able to localize all agents in a MAS (i.e., obtain an estimate of the state $p$ of the MAS). 
Recall that a fully distributed algorithm is an algorithm that can be designed and updated by each agent locally and is resilient to node changes of the MAS. 
Given that the solutions to the localization problem show no significant difference between $h=2$ and $h=3$ in \eqref{eq:single integrator}, we will restrict our subsequent analysis to the case of $h=2$ for brevity.

%\Junfeng{To shuaiting: remove this part?
%Given that there is no significant difference in solutions to the localization problem are 
%similar for $h=2$ and $h=3$ in \eqref{eq:single integrator},  we will restrict our discussion to the case of $h=2$ in the subsequent analysis for brevity.}

\begin{figure}[!t]
\centerline{\includegraphics[width=3.5in]{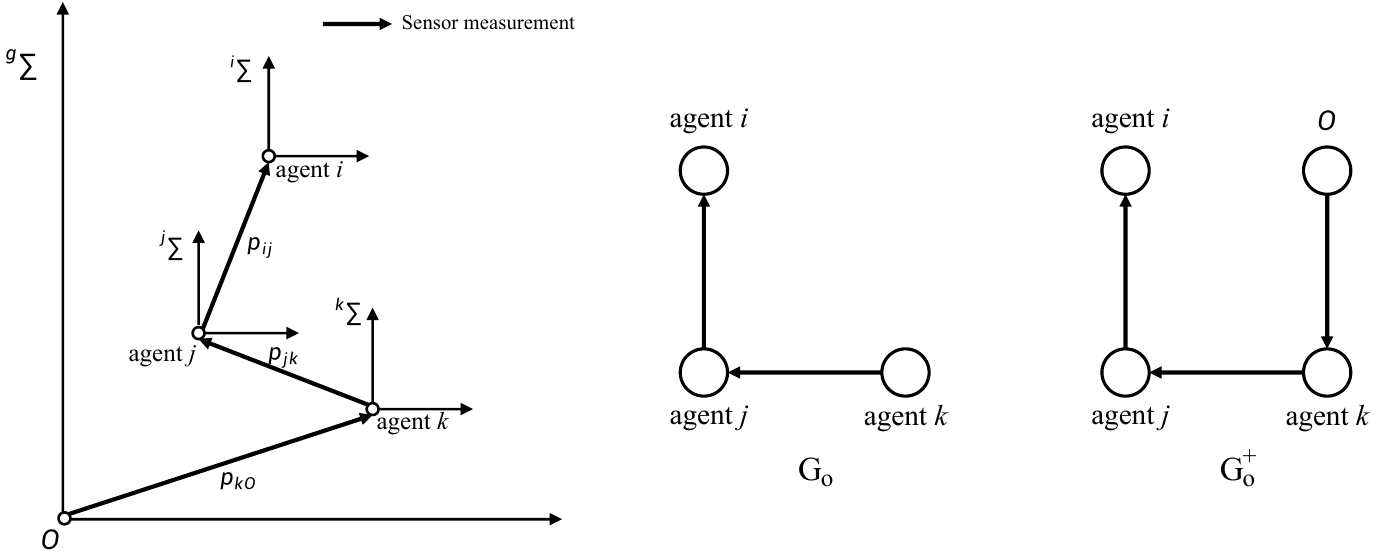}}
\caption{Sensing among agents in the planar case (agents $i$ and $j$ sense the relative position (i.e., $p_{ij}$ and $p_{jk}$), and agent $k$ senses its own absolute position (i.e., $p_{kO}$)) and its corresponding graphs $\mathrm{G}_{\rm o}$ and $\mathrm{G}_{\rm o}^+$.}
\label{fig:setup}
\end{figure}

\subsection{Observability Analysis}\label{sec42}
We define a new directed graph $\mathrm{G}_{\rm o}^+=(\mathcal{V}\cup\{O\},\mathcal{E}_{\rm o}\cup\mathcal{E}_{\rm a})$, where node $O$ denotes the origin of the global reference frame $^g\sum$ and the symbol $\mathcal{E}_{\rm a}$ denotes the set of edges between node $O$ and nodes with absolute position measurements, as shown in Fig.~\ref{fig:setup} where the edge $p_{kO}$ is in set $\mathcal{E}_{\rm a}$.
Assume $|\mathcal{E}_{\rm o}|+|\mathcal{E}_{\rm a}|:=q_{\rm o}+q_{\rm a}=q$, where $q_{\rm o}$ is the number of relative position observations, and $q_{\rm a}$ is the number of absolute position observations.
Then, we permute the observation matrix in \eqref{eq:y-m-agent} and obtain 
$ {C}=\left[\begin{array}{c}
 {C}_{\rm o}\\
 {C}_{\rm a}
\end{array}\right]\in\bb{R}^{2q\times2m},$
where $ {C}_{\rm o}\in\bb{R}^{2q_{\rm o}\times2m}$ corresponds to relative position observations and $ {C}_{\rm a}\in\bb{R}^{2q_{\rm a}\times2m}$ corresponds to absolute position observations.
Note that any edge $(i,j)\in\mathcal{E}_{\rm o}$ corresponds to two rows in $ {C}_{\rm o}$.
Without loss of generality we suppose $i<j$, the two rows in $ {C}_{\rm o}$ associated with the edge $(i,j)$ are of the following form

\begin{footnotesize}
\begin{equation}
\label{eq:Cr}
 {c}_{(i,j)}=\left[
\begin{array}{c|cc|c|cc|c}
 \mathbf{0} & \underset{(2i-1)^{\text{th}}}
{1} & 0 &  \mathbf{0} & \underset{(2j-1)^{\text{th}}}{-1} & 0 &  \mathbf{0}\\
 \mathbf{0} & 0 & \underset{(2i)^{\text{th}}}
{1} &  \mathbf{0} & 0 & \underset{(2j)^{\text{th}}}{-1} &  \mathbf{0}
\end{array}
\right],
\end{equation}
\end{footnotesize}
where $\underset{(2i-1)^{\text{th}}}
{1}$ means that the element $1$ is in column $2i-1$.
Any edge $(O, k)\in\mathcal{E}_{\rm a}$ characterizes the knowledge of the absolute position of agent $k$ and corresponds to two rows in $ {C}_{\rm a}$ of the form

\begin{footnotesize}
\begin{equation}
\label{eq:Ca}
 {c}_{(O, k)}=\left[
\begin{array}{c|cc|c}
 \mathbf{0} & \underset{(2k-1)^{\text{th}}}
{-1} & 0 &  \mathbf{0} \\
 \mathbf{0} & 0 & \underset{(2k)^{\text{th}}}
{-1} &  \mathbf{0}
\end{array}
\right].
\end{equation}
\end{footnotesize}

 % \vspace{-10pt}
A lemma about the rank of the matrix $ {C}_{\rm o}$ is presented as follows.
\begin{lemma}
\label{le:Cr}
The relative measurement matrix $ {C}_{\rm o}$ is not full column rank for any directed graph $\mathrm{G}_{\rm o}^+$.
\end{lemma}

\begin{proof}
By \eqref{eq:Cr}, for any directed graph $\mathrm{G}_{\rm o}^+$, the row sum of $ {C}_{\rm o}$ is equal to $0$, i.e., $ {C}_{\rm o} {1}_{2m}= \mathbf{0}_{2q_{\rm o}\times1}$.
\end{proof}

The observability condition of the single-integrator MAS \eqref{eq:m-agent}-\eqref{eq:y-m-agent} is given as follows.
\begin{proposition}
\label{pro:global obserbeability}
The single-integrator MAS \eqref{eq:m-agent}-\eqref{eq:y-m-agent} is observable if and only if the graph $\overline{\mathrm{G}}_{\rm o}^+$ is connected, where $\overline{\mathrm{G}}_{\rm o}^+$ is the undirected graph obtained by dropping the directions of all edges in $\mathrm{G}_{\rm o}^+$.
%, known as the mirror of $\mathrm{G}_{\rm o}^+$~\cite{you2013consensus},
\end{proposition}

\begin{proof}
The system matrix of \eqref{eq:m-agent} is $ {A}= \mathbf{0}_{2m\times2m}$; the observability matrix can be written as

\begin{footnotesize}
\begin{equation}
\nonumber
{W}_{\rm o}=\left[\begin{array}{c}
 {C}\\
 {CA}\\
\vdots\\
 {CA}^{2m-1}
\end{array}\right]=
\left[\begin{array}{c}
 {C}\\
 \mathbf{0}_{2q\times2m}\\
\vdots\\
 \mathbf{0}_{2q\times2m}
\end{array}\right].
\end{equation}
\end{footnotesize}Thus, to analyze the observability of the MAS \eqref{eq:m-agent}-\eqref{eq:y-m-agent}, we need to discuss the rank of $ {C}$. 

{\bf If.}~Suppose $\overline{\mathrm{G}}_{\rm o}^+$ is connected, we can find a spanning tree $\bar{\mathcal{T}}'_{\rm o}$ from $\overline{\mathrm{G}}_{\rm o}^+$.
We denote the observation matrix corresponding to a spanning tree with $m=|\mathcal{V}|+1$ nodes as $ {C}_m\in\bb{R}^{2(m-1)\times 2(m-1)}$ (there are $2(m-1)$ observations because there are $m-1$ edges in the tree). 
Since $\bar{\mathcal{T}}'_{\rm o}$ consists of a minimal set of edges that connect all nodes of $\overline{\mathrm{G}}_{\rm o}^+$, the column rank of $C$ is no less than the column rank of $C_m$. 
Thus, to prove the MAS is observable, it suffices to prove that $C_m$ is full column rank.
We proceed with the proof by induction.

\underline{Step 1}: Consider the simple case with only two nodes in $\mathrm{G}_{\rm o}^+$, one of them being the global reference frame's origin $O$, i.e., $|\mathcal V|=1$. 
By \eqref{eq:Ca}, it is obvious that $ {C}_2\in\bb{R}^{2\times2}$ is full rank when $\overline{\mathrm{G}}_{\rm o}^+$ is connected.

\underline{Step 2}:  Consider the case with $m$ nodes in $\mathrm{G}_{\rm o}^+$ and one of them being the global reference frame's origin $O$, i.e., $|\mathcal V|=m-1$. 
Assume that the observation matrix $ {C}_m\in\mathbb{R}^{2(m-1)\times 2(m-1)}$ is full rank when $\overline{\mathrm{G}}_{\rm o}^+$ is connected. 
On this basis, we need to prove that the sufficiency extends to the case of $|\mathcal V|=m$.
Two situations need to be discussed.
In the first case, suppose all nodes in $\mathrm{G}_{\rm o}^+$ are directly connected to the origin, i.e., all agents have the absolute position. 
Based on \eqref{eq:Ca}, the observation matrix $ {C}_{m+1}\in\bb{R}^{2m\times 2m}$ corresponding to a spanning tree of $\overline{\mathrm{G}}_{\rm o}^+$ is a negative identity matrix after permutation, thus full column rank. 
In the second case, suppose some nodes are not directly connected to the origin in $\mathrm{G}_{\rm o}^+$, i.e., only a part of the agents can access the absolute position.
Note that a spanning tree of $\overline{\mathrm{G}}_{\rm o}^+$ with $m+1$ nodes becomes a spanning tree of $\overline{\mathrm{G}}_{\rm o}^+$ with $m$ nodes by dropping any leaf node.
Thus, by permutation, the observation matrix $ {C}_{m+1}\in\bb{R}^{2m\times 2m}$ corresponding to a spanning tree of $\mathrm{G}_{\rm o}^+$ can be written as  
     ${C_{m+1}}=   
    \left[
    \begin{array}{ccc|cc}
    \multicolumn{3}{c|}{{C}_{m}}& 0 &  0\\
    \hline
     \mathbf{0} & 1    & 0 & -1   & 0\\
     \mathbf{0} & 0    & 1 & 0    & -1  \\
    \end{array}
    \right].$
Recall the assumption that $ {C}_{m}$ is full column rank at the beginning of step 2, and it follows that $ {C}_{m+1}$ is full column rank.

{\bf Only if.}~ 
Suppose that the undirected graph $\overline{\mathrm{G}}_{\rm o}^+$ is not connected.
We can find a maximal connected subgraph containing the origin $O$, denoted by $\overline{\mathrm{G}}'_{\rm o1}$.
Then we separate $\overline{\mathrm{G}}_{\rm o}^+$ into two subgraphs, i.e.,  $\overline{\mathrm{G}}'_{\rm o1}$ and another subgraph denoted by $\overline{\mathrm{G}}'_{\rm o2}$.
Note that $\overline{\mathrm{G}}'_{\rm o1}$ is connected while $\overline{\mathrm{G}}'_{\rm o2}$ is not necessarily connected.
According to the graphs $\overline{\mathrm{G}}'_{\rm o1}$ and $\overline{\mathrm{G}}'_{\rm o2}$, by permutation, the observation matrix $ {C}$ can be written as
    $C=\left[\begin{array}{c|c}
    C_{\rm o1} &  \mathbf{0}\\
    \hline
     \mathbf{0} &  C_{\rm o2}
    \end{array}\right].
    $
The submatrix $ {C}_{\rm o2}$ can be regarded as the corresponding relative measurement matrix of the directed graph, which induces $\overline{\mathrm{G}}'_{\rm o2}$.
By Lemma~\ref{le:Cr}, we obtain $ {C}_{\rm o2}$ is not full column rank. 
Hence, $ {C}$ is not full column rank, i.e., the system \eqref{eq:m-agent}-\eqref{eq:y-m-agent} is unobservable.
\end{proof}

We further discuss the observability condition of the pair $(A_{ii}, C_{ii})$ of the single-integrator modeled agent \eqref{eq:single integrator}-\eqref{eq:single-yi}.
\begin{proposition}
\label{pro:local obserbeability}
In a single-integrator MAS \eqref{eq:m-agent}-\eqref{eq:y-m-agent} consisting of $m$ agents with dynamics \eqref{eq:single integrator}-\eqref{eq:single-yi},  for $\forall i\in\mathcal V$, the pair $( {A}_{ii}, {C}_{ii})$ is observable if and only if node $i$ is not a source in graph $\mathrm{G}_{\rm o}^+$, i.e., agent $i$ has at least one relative measurement or absolute measurement.
\end{proposition}
\begin{proof}
The proof follows from the fact that for a source $i$, $C_{ii}$ is an all-zeros matrix, otherwise ${C}_{ii}$ is a vector whose elements are the matrix $-I_2$.
\end{proof}

\subsection{Fully Distributed Localization Scheme}\label{sec43}
The sensing graph $\mathrm G_{\rm o}$ is frequently not a directed acyclic graph (DAG) in the cooperative localization problem.
To overcome this challenge, we devise a distributed algorithm, outlined in Algorithm 1, to construct a directed acyclic sensing graph $\mathrm G_{\rm o}^*$ from a general graph $\mathrm G_{\rm o}$.

In Algorithm~\ref{alg:case2}, each agent first randomly chooses a positive integer as its unique ID. 
Then, we introduce an \textit{agent hierarchy rule}, whereby all agents with absolute positions are at layer $0$, and agents with relative positions to agents at layer $0$ are at layer $1$. 
We continue this procedure until each agent is assigned a layer. 
The function $\rm layerjudge()$ is defined to characterize the hierarchy rule.
With the above ID selection and hierarchical operation, we stipulate that the relative position measurement between two agents at different layers can only be used by the agent at the larger layer,
while the relative position measurement between two agents at the same layer can only be used by the agent with the larger ID.
This process is achieved by agents transmitting relative positions through $\mathrm G_{\rm c}$, given that the communication range of an agent is usually wider than its sensing range, i.e., $\mathrm G_{\rm o}\subset \mathrm G_{\rm c}$.

The proposed DAGC algorithm is fully distributed, as each agent can determine its ID and identify its layer independently through local measurement.
Furthermore, the proposed DAGC algorithm does not affect the observability condition of the MAS \eqref{eq:m-agent}-\eqref{eq:y-m-agent}.
Notably, since the construction of $\mathrm G_{\rm o}^*$ only involves changing the direction of the edges of $\mathrm G_{\rm o}$ and not the number of edges, if the origin sensing graph $\mathrm G_{\rm o}$ satisfies Propositions~\ref{pro:global obserbeability} and~\ref{pro:local obserbeability}, so does the constructed $\mathrm G_{\rm o}^*$. 
We use Example~\ref{exam:DAG} to illustrate the outcome of the DAGC algorithm.
\begin{example}
\label{exam:DAG}
Consider a sensing graph $\mathrm G_{\rm o}$ shown in Fig.~\ref{fig:DAGC}.
Agent $i$ is the only one that has its absolute position.
The graph $\mathrm G_{\rm o}$ is not a DAG.
By Algorithm~\ref{alg:case2}, $\mathrm G_{\rm o}$ can be converted to a directed acyclic graph $\mathrm G_{\rm o}^*$ with the IDs and layers shown in Fig.~\ref{fig:DAGC}.
\begin{figure}[!ht]
\centerline{\includegraphics[width=7cm]{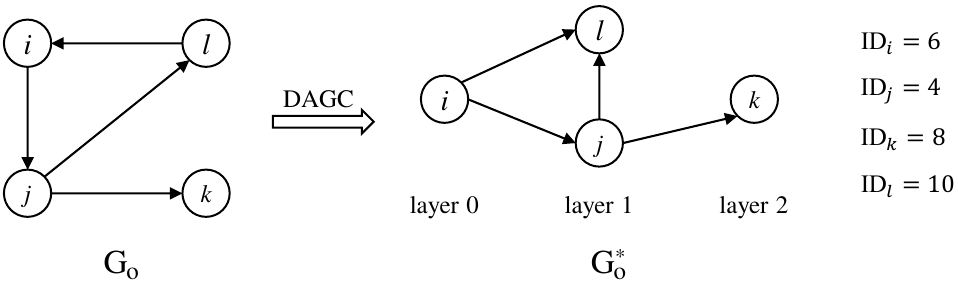}}
\caption{An example of converting a sensing graph into an acyclic graph via Algorithm~\ref{alg:case2}.}
\label{fig:DAGC}
\end{figure}
\end{example}

We now present the form of the fully distributed cooperative localization algorithm for the single-integrator multi-agent system \eqref{eq:m-agent}-\eqref{eq:y-m-agent} as follows:

\begin{footnotesize}
\begin{subequations}
\label{eq:localization}
\begin{align}
& {\dot{\hat{p}}}_j^{(i)} = {u}_{j}+ \sum_{l\in \mathcal N_i(\mathrm{G}_{\rm c})} w_{il}^{(j)}({\hat{p}}_{j}^{(l)}- {\hat{p}}_j^{(i)}), {\rm ~for~}j\in \mathcal{V}/\{i\};
     \label{eq:L3}\\
& {\dot{\hat{p}}}_i^{(i)}= {u}_{i}+w_{i0}^{(i)}  ({\bar{p}}_{i}- {\hat{p}}_i^{(i)})+ \sum_{l\in \mathcal N_i(\mathrm{G}_{\rm c})} w_{il}^{(i)}({\hat{p}}_{i}^{(l)}- {\hat{p}}_i^{(i)}),
   \label{eq:L2}
\end{align}
with
\begin{align}
& {\dot{\bar{p}}}_i= {u}_{i}+ {F}_i({y}_i^*-({C}_{ii} {\bar{p}}_i+\sum_{l\in\mathcal N_i(\mathrm{G}_{\rm o}^*)} {C}_{il} {\hat{p}}_l^{(i)})),
 \label{eq:L1}
\end{align}
\end{subequations}
\end{footnotesize}where ${y}_i^*$ denotes the measurement of agent $i$ after adjustment by the DAGC algorithm, and all other notations are analogous to those in \eqref{eq:Do}. 
The Luenberger gain $F_i$ and the consensus gains $w_{ij}$ are chosen by utilizing the same approach in 
Section~\ref{sec32}.

\begin{algorithm}
 \caption{Distributed DAG construction (DAGC) from \textit{Agent Hierarchy Rule}.}
 \label{alg:case2}
 \footnotesize
 \begin{algorithmic}[1]
  \State {\bf for} agent $i\in \mathcal V$ {\bf do}
  \State ~~~~${\rm ID}_i={\rm random}()\in \bb{N}^{+}$ (i.e., the set of positive integers)
  \State ~~~~{\bf for} agent $j\in\mathcal N_i(\mathrm G_{\rm o})$ {\bf do}
  \State ~~~~~~~~{\bf if} ${\rm layerjudge}(i)<{\rm layerjudge}(j)$ {\bf then} $ {p}_{ji}\gets -{p}_{ij}$, $ {p}_{ij}\gets  \mathbf{0}$
  \State ~~~~~~~~{\bf elseif} ${\rm layerjudge}(i)>{\rm layerjudge}(j)$ {\bf then} 
  \State ~~~~~~~~~~~~~~~$ {p}_{ij}\gets {p}_{ij}$, $ {p}_{ji}\gets  \mathbf{0}$
  \State ~~~~~~~~{\bf else}
  \State ~~~~~~~~~~~~~{\bf if} ${\rm ID}_i={\rm ID}_j$ {\bf then} ${\rm ID}_j\gets {\rm ID}_j+1$
  \State ~~~~~~~~~~~~~{\bf elseif} ${\rm ID}_i<{\rm ID}_j$ {\bf then} $ {p}_{ji}\gets -{p}_{ij}$, $ {p}_{ij}\gets  \mathbf{0}$
  \State ~~~~~~~~~~~~~{\bf else} $ {p}_{ij}\gets {p}_{ij}$, $ {p}_{ji}\gets  \mathbf{0}$
  \State ~~~~~~~~~~~~~{\bf end if}
  \State ~~~~~~~~{\bf end if}
  \State ~~~~{\bf end for}
  \State ~~~~{\bf for} agent $j\in\mathcal S_i(\mathrm G_{\rm o})$ {\bf do}
  \State ~~~~~~~~{\bf if} ${\rm layerjudge}(j)<{\rm layerjudge}(i)$ {\bf then} $ {p}_{ij}\gets -{p}_{ji}$, $ {p}_{ji}\gets  \mathbf{0}$
  \State ~~~~~~~~{\bf elseif} ${\rm layerjudge}(j)>{\rm layerjudge}(i)$ {\bf then} 
  \State ~~~~~~~~~~~~~~~$ {p}_{ji}\gets -{p}_{ji}$, $ {p}_{ij}\gets  \mathbf{0}$
  \State ~~~~~~~~{\bf else}
  \State ~~~~~~~~~~~~~{\bf if} ${\rm ID}_i={\rm ID}_j$ {\bf then} ${\rm ID}_j\gets {\rm ID}_j+1$
  \State ~~~~~~~~~~~~~{\bf elseif} ${\rm ID}_j<{\rm ID}_i$ {\bf then} $ {p}_{ij}\gets -{p}_{ji}$, $ {p}_{ji}\gets  \mathbf{0}$
  \State ~~~~~~~~~~~~~{\bf else} $ {p}_{ji}\gets {p}_{ji}$, $ {p}_{ij}\gets  \mathbf{0}$
  \State ~~~~~~~~~~~~~{\bf end if}
  \State ~~~~~~~~{\bf end if}
  \State ~~~~{\bf end for}
  \State {\bf end for}
 \end{algorithmic}
\end{algorithm}

The following theorem provides conditions for the convergence of $ {\hat{p}}^{(i)}$ to ${p}$ when ${u}$ is known to all agents, where $ {\hat{p}}^{(i)}=[(\hat{p}_1^{(i)})^\top, \ldots,  (\hat{p}_m^{(i)})^\top]^\top$ denotes the estimate of $p$ generated by agent $i$. 
For real-world applications where an agent cannot access the inputs of other agents and where the inputs are assumed to be bounded, we can demonstrate that the estimation errors of our proposed localization algorithm remain bounded according to input-to-state stability analysis.
\begin{theorem}
\label{pro:main}
Consider a single-integrator multi-agent system \eqref{eq:m-agent}-\eqref{eq:y-m-agent}, provided that ${u}$ is known to all agents, for $\forall i\in\mathcal V$, the estimate $ {\hat{p}}^{(i)}$ obtained by agent $i$ using \eqref{eq:localization} asymptotically converges to $ {p}$ if and only if all the following conditions are satisfied:
\begin{enumerate}
\item[(1)]$\overline{\mathrm{G}}_{\rm o}^+$ is connected;
\item[(2)]$\mathrm{G}_{\rm o}^+$ has no sources except the origin $O$;
\item[(3)]$\mathrm{G}_{\rm c}$ is strongly connected.
\end{enumerate}
\end{theorem}
\begin{proof}
The result follows from Theorem~\ref{theorem:main1} and Propositions~\ref{pro:global obserbeability} and~\ref{pro:local obserbeability}.
% The detail is omitted for brevity.
\end{proof}

% \begin{remark}
% \label{re:sectionrelationship}
% The cooperative localization algorithm \eqref{eq:localization} can be seen as a special case of the proposed distributed observer \eqref{eq:Do}, where the system matrices and state-coupling matrices of each agent are $A_{ii}=A_{ij}=\mathbf{0}$ (from which we know that $\max_{i\in\mathcal V} \rho(A_{ii})=0$)
% for all $i,j\in\mathcal V$, and the coupling gain can be set to be $\mu>0$
% according to~\eqref{eq:mu}.
% We set $\mu=1$ in our localization algorithm \eqref{eq:localization} for brevity. Any positive coupling gain applies equally to the cooperative localization algorithm~\eqref{eq:double localization} for
% double-integrator MASs.
% \end{remark}

\subsection{Extension to Double-integrator Multi-agent Systems}\label{sec44}
In cooperative localization problems, we may also be interested in estimating mobile nodes' velocities, which requires modeling each agent's kinematics as a double-integrator system.
The above results can be extended to double-integrator multi-agent systems under a similar setting, where all local reference frames are aligned to the global reference frame, and all agents can sense the relative position and relative velocity, but not all agents can sense the absolute position and absolute velocity.
The notations in this part are analogous to those in the last subsection.  
In particular, consider a double-integrator multi-agent system consisting of $m$ agents modeled by  
\begin{equation}
\label{eq:double integrator}
\left\{\begin{array}{l}
{\dot{p}}_i= {v}_i,\\
{\dot{v}}_i= {u}_i,
\end{array}
\right. i\in\mathcal V
\end{equation}
where $p_i\in\bb{R}^{h}$, $v_i\in\bb{R}^{h}$ and $u_i\in\bb{R}^{h}$ are the position, velocity and control input of agent $i$, respectively.
The sensor measurement of agent $i$ can be written as 
\begin{equation}
\label{eq:double-yi}
 {y}_i= {C}_{ii} {s}_i+\sum_{j\in\mathcal N_i(\mathrm G_{\rm o})} {C}_{ij} {s}_j
\end{equation}
where $s_i=[p_i^\top,v_i^\top]^\top\in\bb{R}^{2h}$, ${y}_i\in\bb{R}^{q_i}$, ${C}_{ii}\in\bb{R}^{q_i\times 2h}$ and $ {C}_{ij}\in\bb{R}^{q_i\times 2h}$ are the state, output, observation matrix and output coupling matrix of agent $i$, respectively, and the matrices $C_{ii}$ and $C_{ij}$ are only composed of the matrices $I_{2h}$, $-I_{2h}$ and $\mathbf 0_{2h\times 2h}$. 
Thus, the entire state of the double-integrator multi-agent system can be written as $s=[s_1^\top,\dots,s_m^\top]^\top$, and its dynamics can be written as 
\begin{align}
&\dot{s}=As+Bu \label{eq:double-agent}\\
&y=Cs \label{eq:y-double-agent}
\end{align}
where, for $i\in \mathcal V$, the system matrix is 
$A=\textup{diag}(J_1,\ldots, J_m)$ with $J_i=J=\left[\mathbf{0}_4, \mathbf{0}_4,b_1,b_2\right]\in\bb{R}^{4\times 4}$
and the input matrix is 
$B=\textup{diag}(K_1,\ldots, K_m) $ with $K_i=K=\left[b_2,b_4\right]\in\bb{R}^{4\times 2}$.
Note that the system matrix $A$ satisfies ${A}^2=\mathbf{0}$.
To analyze the observability of the double-integrator multi-agent system \eqref{eq:double-agent}-\eqref{eq:y-double-agent}, we need to discuss the rank of the following observability matrix 
$ {W}_{\rm o}=
\left[\begin{array}{c}
 {C}\\
 {CA}\\
 \mathbf{0}
\end{array}\right].$
We can get the observability results similar to Propositions~\ref{pro:global obserbeability} and \ref{pro:local obserbeability}.

% By the similar , We obtain the following results by an analysis similar to Section~\ref{sec42}.
% \begin{corollary}
% \label{coro:2-global observability}
% The double-integrator multi-agent system \eqref{eq:double-agent}-\eqref{eq:y-double-agent} is observable if and only if the graph $\overline{\mathrm{G}}_{\rm o}^+$ is connected, where $\overline{\mathrm{G}}_{\rm o}^+$ is the undirected graph obtained by dropping the direction of all edges in $\mathrm{G}_{\rm o}^+$.
% \end{corollary}
% % \noindent \textbf{Proof.~}
% % The proof is similar to that of Proposition~\ref{pro:global obserbeability}.
% % The detail is omitted for brevity.
% % \hfill$\blacksquare$

% \begin{corollary}
% \label{coro:2-local observability}
% In a double integrator multi-agent system \eqref{eq:double-agent}-\eqref{eq:y-double-agent} consisting of $m$ agents with dynamics \eqref{eq:double integrator}-\eqref{eq:double-yi}, for $\forall i\in\mathcal V$, the pair $( {A}_{ii}, {C}_{ii})$ is observable if and only if node $i$ is not a source in graph $\mathrm{G}_{\rm o}^+$, i.e, agent $i$ has at least one observation.
% \end{corollary}
% \noindent \textbf{Proof.~}
% The proof is the same as that of Proposition~\ref{pro:local obserbeability}.
% The detail is omitted for brevity.
% \hfill$\blacksquare$

The form of fully distributed cooperative localization algorithm for the double-integrator MAS \eqref{eq:double-agent}-\eqref{eq:y-double-agent} is given by

\begin{footnotesize}
\begin{subequations}
\label{eq:double localization}
\begin{align}
{\dot{\hat{s}}}_j^{(i)} =& J\hat{s}_j^{(i)}+K{u}_{j}\nonumber\\
     &+ \sum_{l\in \mathcal N_i(\mathrm{G}_{\rm c})} w_{il}^{(j)}({\hat{s}}_{j}^{(l)}- {\hat{s}}_j^{(i)}), {\rm ~for~}j\in \mathcal{V}/\{i\};
     \label{eq:D3}\\
{\dot{\hat{s}}}_i^{(i)}=& J\hat{s}_i^{(i)}+K{u}_{i}\nonumber\\
  &+w_{i0}^{(i)}({\bar{s}}_{i}- {\hat{s}}_i^{(i)})+ \sum_{l\in \mathcal N_i(\mathrm{G}_{\rm c})} w_{il}^{(i)}({\hat{s}}_{i}^{(l)}- {\hat{s}}_i^{(i)}), 
   \label{eq:D2}
\end{align}
with
\begin{align}
 {\dot{\bar{s}}}_i=& J\bar{s}+K{u}_{i}+ {F}_i({y}_i^*-(C_{ii}\bar{s}_i+\sum_{l\in\mathcal N_i(\mathrm{G}_{\rm o}^*)}C_{il} {\hat{s}}_i^{(l)})). 
 \label{eq:D1}
\end{align}
\end{subequations}
\end{footnotesize}

% \vspace{-10pt}
The following corollary provides conditions for the convergence of $\hat{s}^{(i)}$ to $s$, where $ \hat{s}^{(i)}=[(\hat{s}_1^{(i)})^\top, \ldots, (\hat{s}_m^{(i)})^\top]^\top$ is the estimate of $s$ generated by agent $i$.
\begin{corollary}
\label{coro:main}
Consider a double-integrator MAS \eqref{eq:double-agent}-\eqref{eq:y-double-agent}, provided that ${u}$ is known to all agents, for $\forall i\in\mathcal V$, the estimate $ {\hat{s}}^{(i)}$ obtained by agent $i$ using \eqref{eq:double localization} asymptotically converges to $s$ if and only if all the following conditions are satisfied:
\begin{itemize}
\item[(1)]$\overline{\mathrm{G}}_{\rm o}^+$ is connected;
\item[(2)]$\mathrm{G}_{\rm o}^+$ has no sources except the origin $O$;
\item[(3)]$\mathrm{G}_{\rm c}$ is strongly connected. 
\end{itemize}
\end{corollary}
% \noindent \textbf{Proof.~}
% The proof is the same as that of Theorem~\ref{pro:main}.
% The detail is omitted for brevity.
% \hfill$\blacksquare$

\begin{remark}
\label{re:u+fully}
The parameters in distributed localization algorithms \eqref{eq:localization} and \eqref{eq:double localization} are determined locally by each single agent, and their updates depend solely on local measurement and information exchanged with neighbors in $\mathrm G_{\rm c}$ (i.e., via local communications). 
Therefore, the proposed distributed localization algorithms \eqref{eq:localization} and \eqref{eq:double localization} can be operated in a fully distributed manner. 
\end{remark}

\begin{remark}
\label{re:n-order intergrator dynamics}
The above observability analysis results and localization algorithms for single- and double-integrator MASs can be extended to high-order integrator MASs. The specific analysis process and results are similar to those of single-integrator MAS. 
\end{remark}

\section{Numerical Examples}\label{sec5}
In this section, we present simulation results for a general multi-agent system consisting of three agents and a single-integrator multi-agent system consisting of six agents. 

\subsection{Estimation for a General MAS}\label{sec52}
Consider a MAS described by \eqref{eq:sys1}-\eqref{eq:sys2} with three agents.
The system matrices, observation matrices and coupling matrices are given as follows: 
${A}_{11}=\left[\begin{array}{cc}
1.2 & 1 \\
0 & 0.8 
\end{array}\right],\ \  {A}_{22}=1.03,\ \  {A}_{33}=0.3, 
{C}_{11}=\left[\begin{array}{cc}
1 & 0 \\
0 & 1 
\end{array}\right],\ \  {C}_{22}= {C}_{33}=1,
{A}_{21}=\left[\begin{array}{cc}
0.8 & 1 
\end{array}\right],\ \  {C}_{21}=\left[\begin{array}{cc}
0.8 & 1.2 
\end{array}\right].$
The communication graph is shown in Fig.~\ref{fig:communicationGcinout}(a).
It can be seen that the pairs $({A},{C})$ and $( {A}_{ii}, {C}_{ii})$ for each $i\in\{1,2,3\}$ are observable, and the graphs $\mathrm{G}_{\rm s}$ and $\mathrm{G}_{\rm o}$ are identical directed acyclic graphs.
Hence, the proposed observer \eqref{eq:Do} is applicable to this setting. 
Provided that the input $u$ is known to all agents, by Theorem~\ref{theorem:main1}, the local Luenberger gains are chosen as 
${F}_1=\left[\begin{array}{cc}
4.2 & 0\\
0 & 4.8
\end{array}\right],~~ {F}_2=3,~~ {F}_3=2,$
$\mu$ is set to $10$, and the consensus gains $w_{i0}^{(i)}, w_{il}^{(j)}$ are all set to $1$ for all $i,j\in \mathcal{V}, l\in \mathcal N_i(\mathrm G_{\rm c})$.
The initial state is set to $ {x}(0)=[0.5, -0.5, 0.5, 0.5]^\top$, and the initial estimates for all agents are $ \mathbf{0}$.
Simulation results, shown in Fig.~\ref{fig:heteroMAS}, indicate that the estimation error dynamics asymptotically converge to zero, as expected from Theorem~\ref{theorem:main1}, where $^f {e}_j^{(i)}$ denotes the estimation error of the first state of agent $j$ obtained by agent $i$.

\begin{figure}[!ht]
\begin{minipage}{0.25\linewidth}
\centerline{\includegraphics[width=2.3cm]{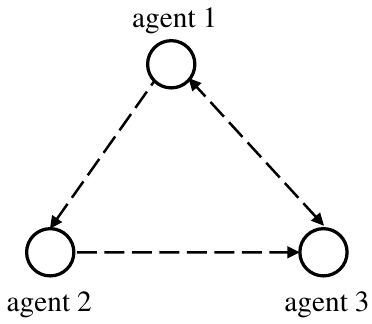}}
  \centerline{\scriptsize(a)}
\end{minipage}
\hfill
\begin{minipage}{0.33\linewidth}
  \centerline{\includegraphics[width=3.2cm]{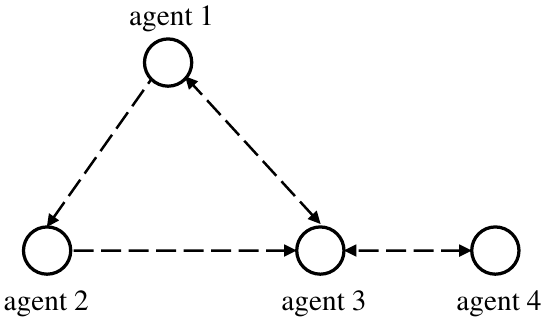}}
  \centerline{\scriptsize(b)}
\end{minipage}
\hfill
\begin{minipage}{0.25\linewidth}
  \centerline{\includegraphics[width=1.3cm]{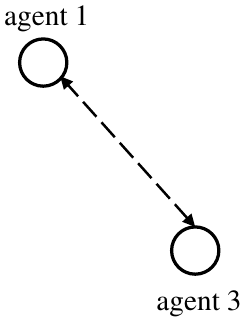}}
  \centerline{\scriptsize(c)}
\end{minipage}
\caption{Communication graphs in Subsection~\ref{sec52}.}
\label{fig:communicationGcinout}
\end{figure}

\begin{figure}[!ht]
\begin{minipage}{0.48\linewidth}
\centerline{\includegraphics[width=4cm]{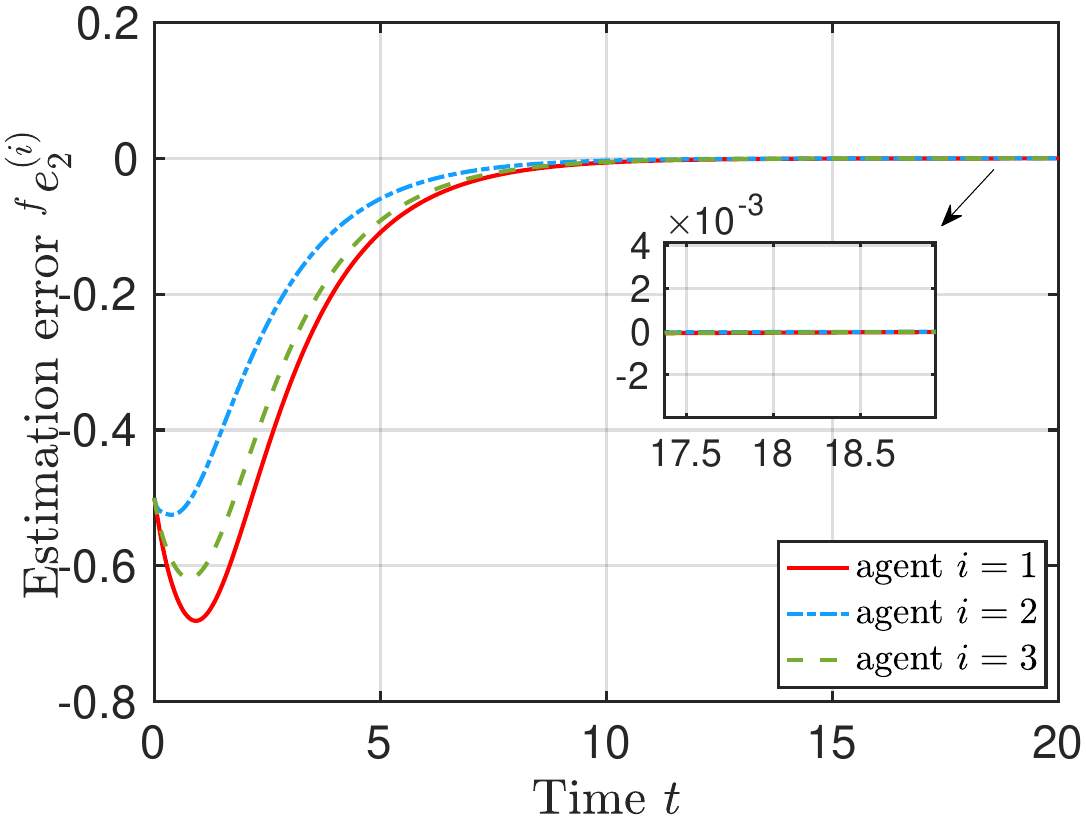}}
  \centerline{\scriptsize(a)}
\end{minipage}
\hfill
\begin{minipage}{0.48\linewidth}
  \centerline{\includegraphics[width=4cm]{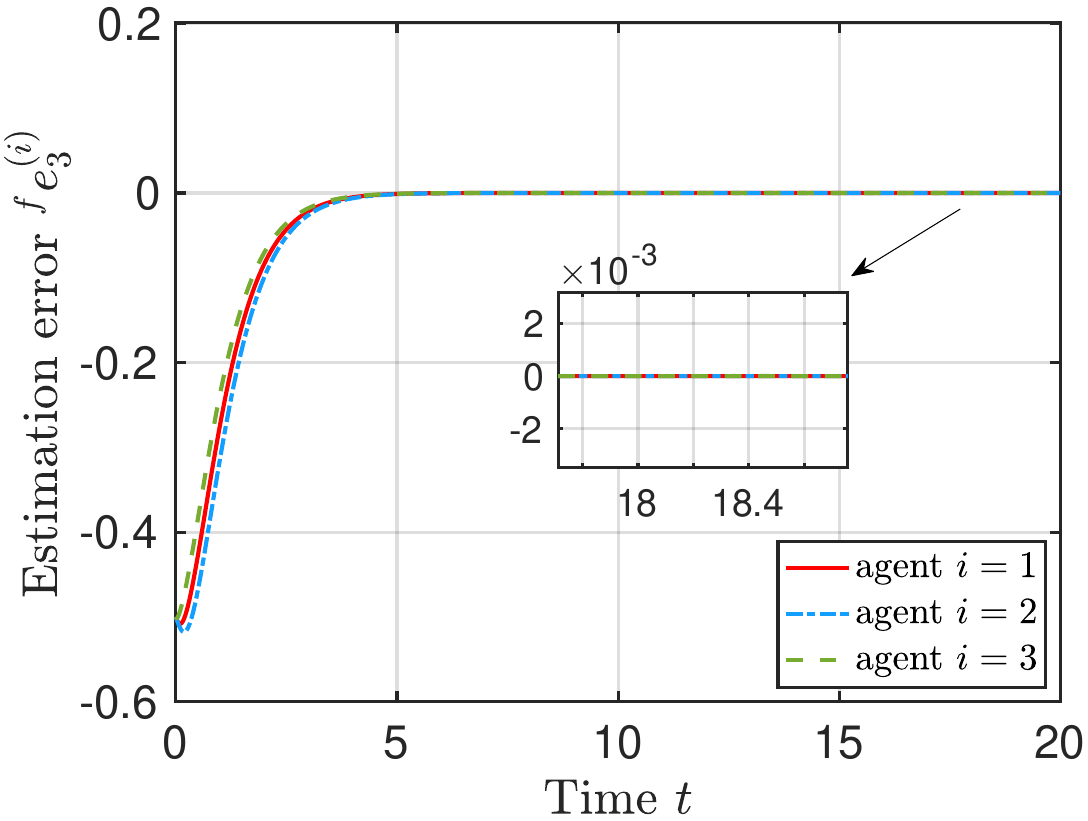}}
  \centerline{\scriptsize(b)}
\end{minipage}
\caption{Estimation error dynamics: (a) the dynamics of $^f {e}_2^{(i)}$; (b) the dynamics of $^f {e}_3^{(i)}$.}
\label{fig:heteroMAS}
\end{figure}

We proceed by presenting a simulation to illustrate the performance of the proposed distributed observer in the presence of noise.
The MAS consists of three agents, each with process and measurement noise bounded by $0.05$ in their state-space models.
The only difference from the previous simulation is the introduction of bounded noise.
The simulation results in Fig.~\ref{fig:heteroMASnoise} demonstrate that the estimation error dynamics are bounded, confirming the effectiveness of the proposed algorithm against noise.

\begin{figure}
\begin{minipage}{0.48\linewidth}
\centerline{\includegraphics[width=4cm]{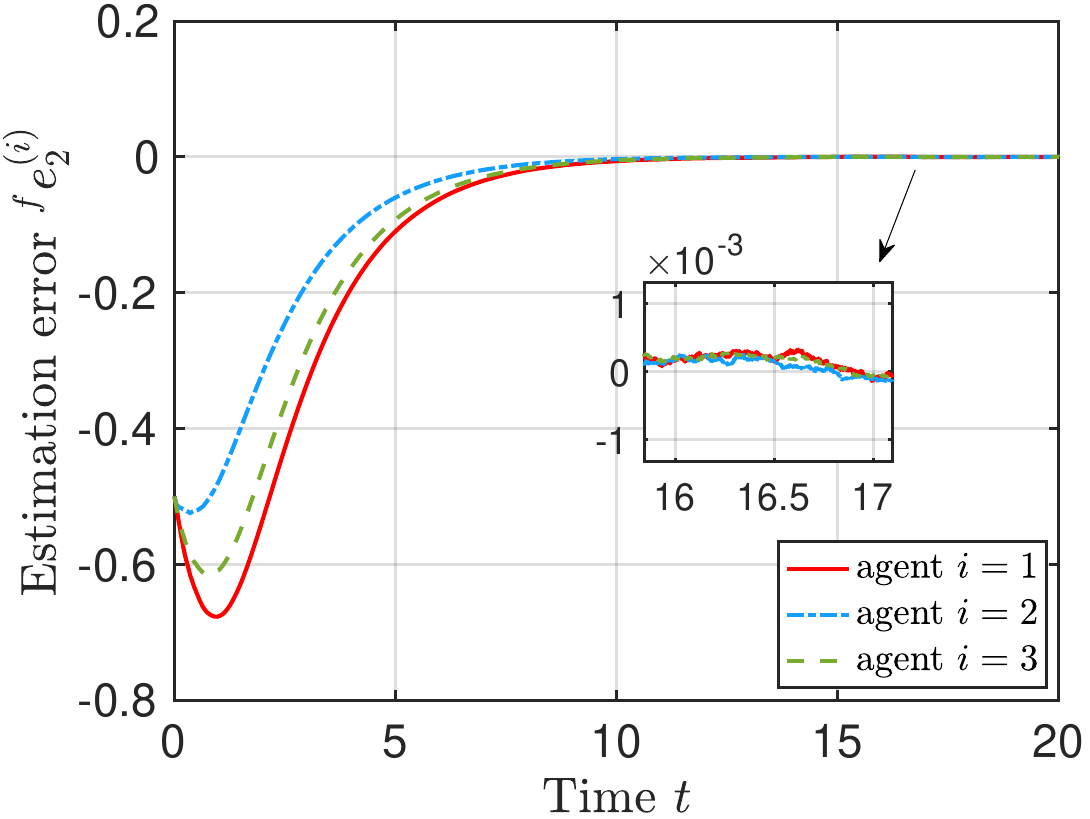}}
  \centerline{\scriptsize(a)}
\end{minipage}
\hfill
\begin{minipage}{0.48\linewidth}
  \centerline{\includegraphics[width=4cm]{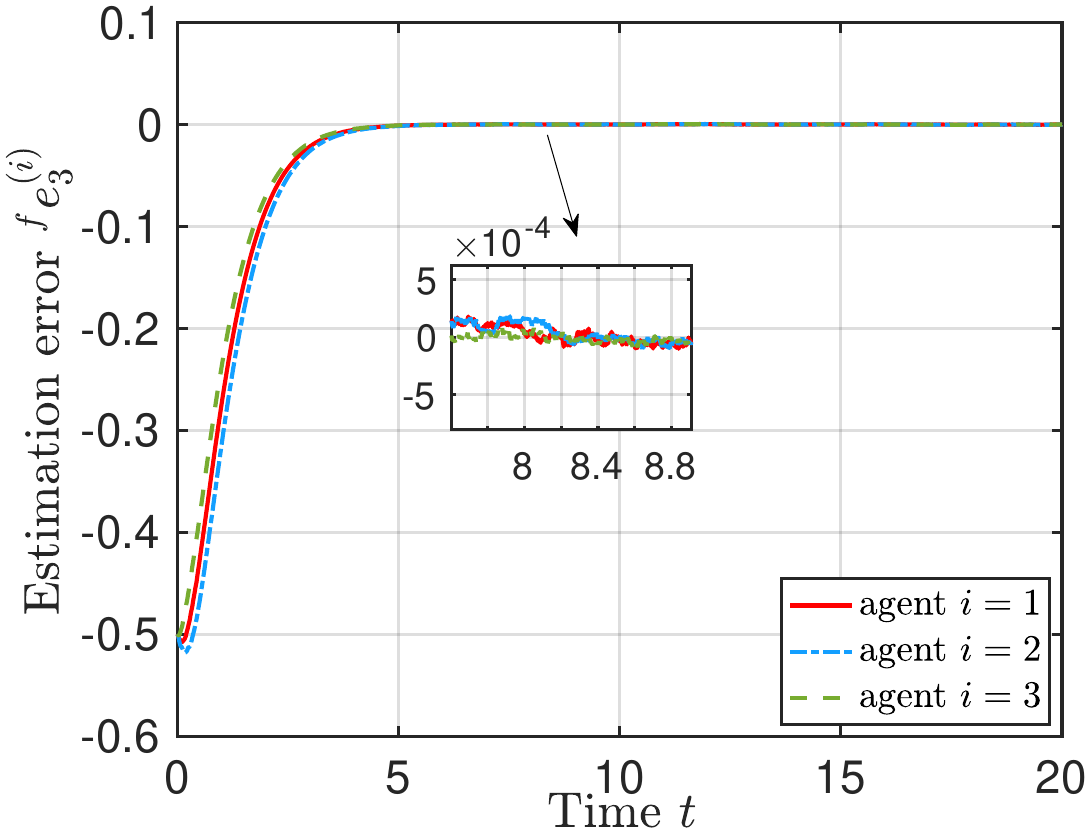}}
  \centerline{\scriptsize(b)}
\end{minipage}
\caption{Performance in the presence of noise: (a) the dynamics of $^f {e}_2^{(i)}$; (b) the dynamics of $^f {e}_3^{(i)}$.}
\label{fig:heteroMASnoise}
\end{figure}

\begin{figure}
\begin{minipage}{0.48\linewidth}
  \centerline{\includegraphics[width=4cm]{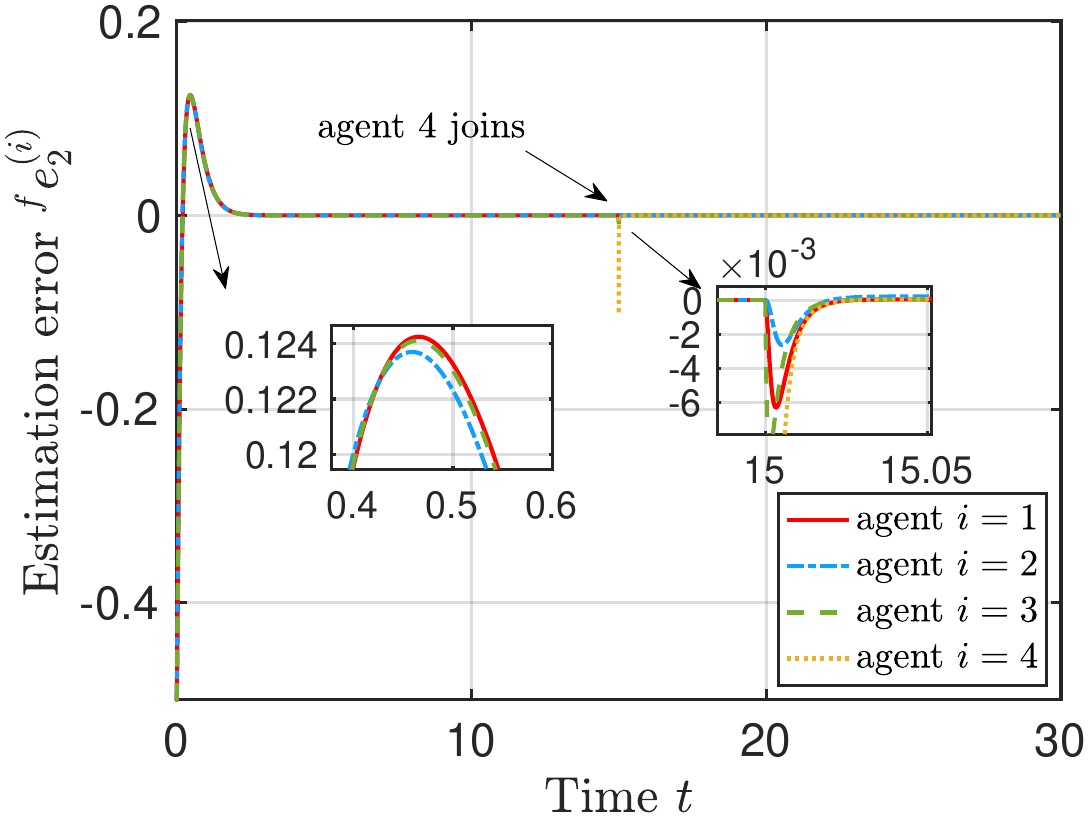}}
  \centerline{\scriptsize(a)}\label{fig:merge2}
\end{minipage}
\hfill
\begin{minipage}{0.48\linewidth}
\centerline{\includegraphics[width=4cm]{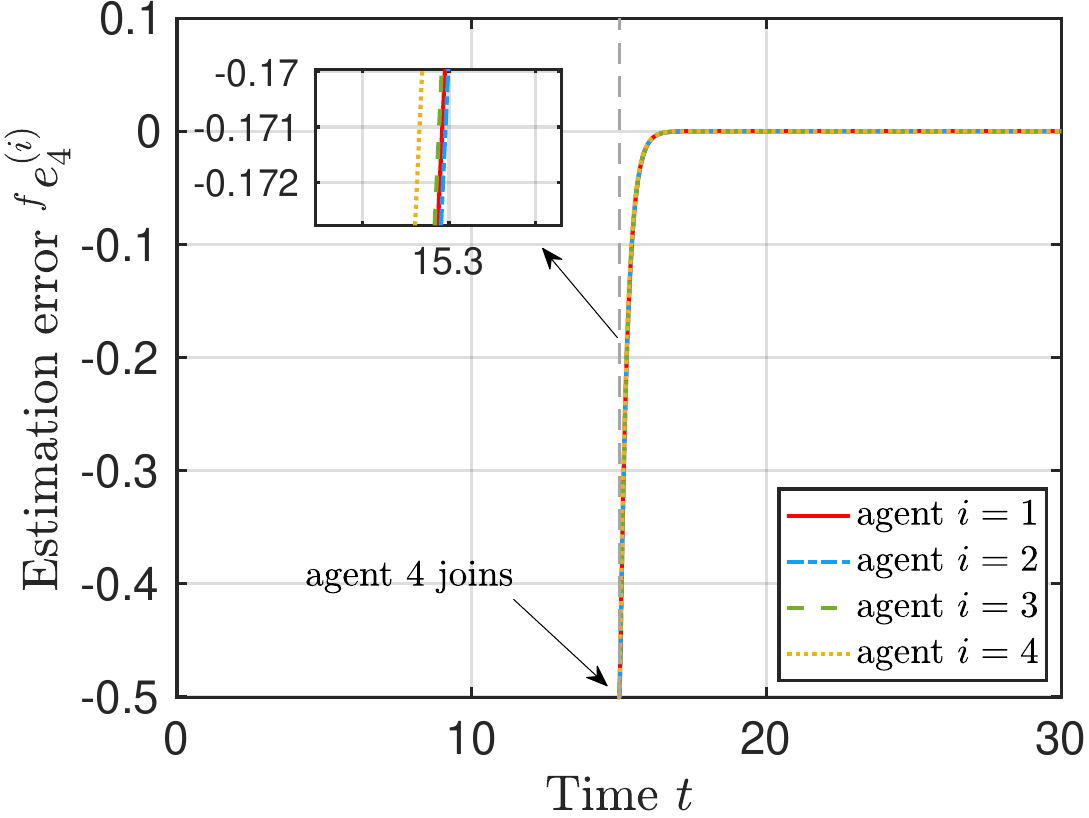}}
  \centerline{\scriptsize(b)}\label{fig:merge4}
\end{minipage}
\caption{Performance in agent joining scenario where agent $4$ joins at $t=15$: (a) the dynamics of $^f {e}_2^{(i)}$; (b) the dynamics of $^f {e}_4^{(i)}$.}
\label{fig:plugin}
\end{figure}

\begin{figure}
\begin{minipage}{0.48\linewidth}
\centerline{\includegraphics[width=4cm]{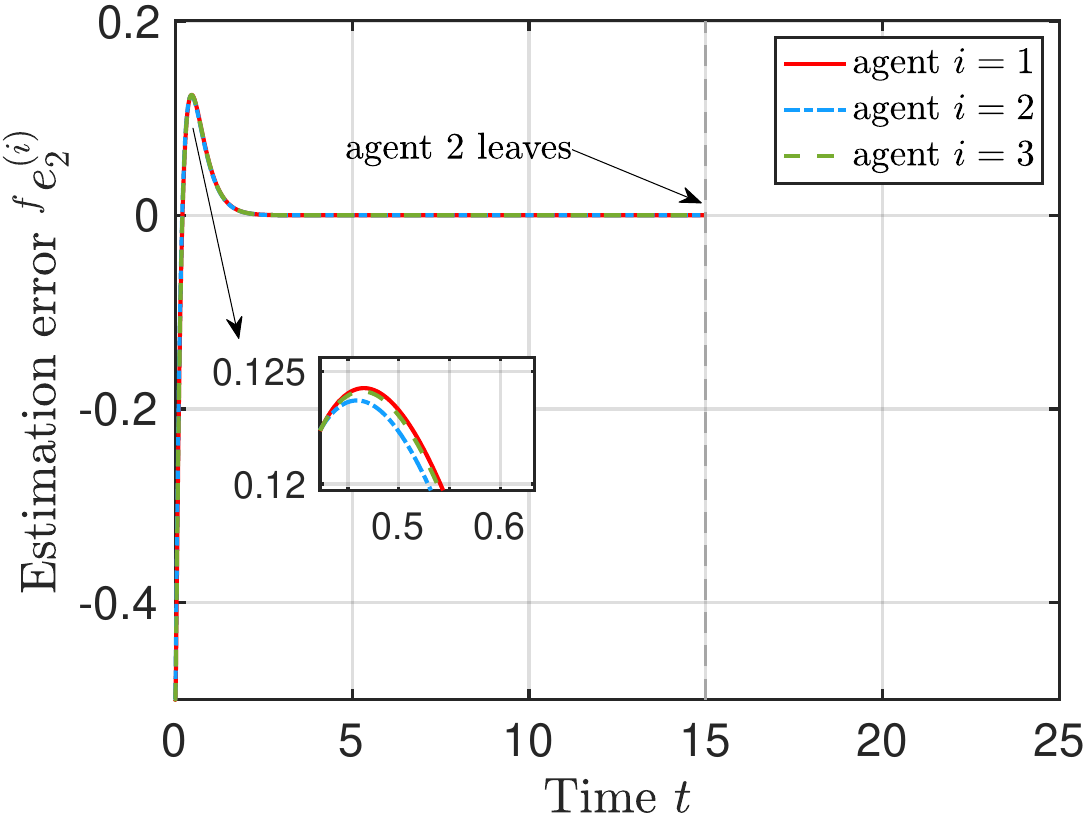}}
  \centerline{\scriptsize(a)}\label{fig:split2}
\end{minipage}
\hfill
\begin{minipage}{0.48\linewidth}
\centerline{\includegraphics[width=4cm]{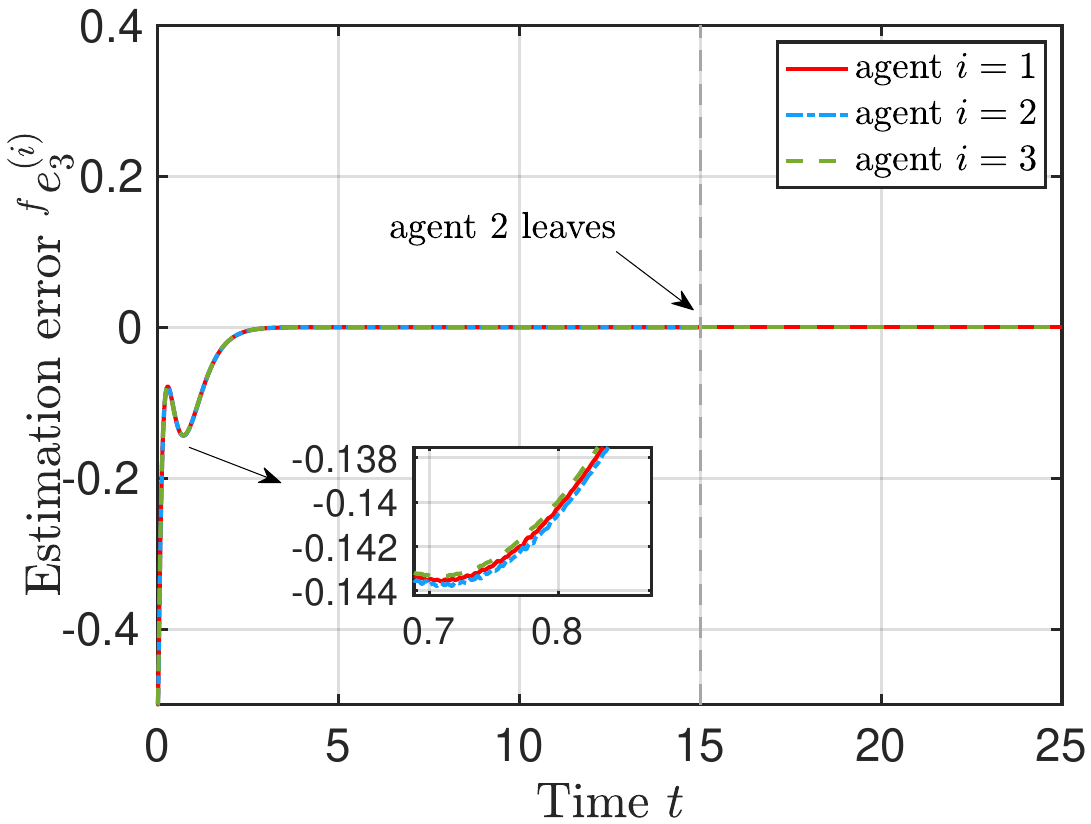}}
  \centerline{\scriptsize(b)}\label{fig:split3}
\end{minipage}
\caption{Performance in agent leaving scenario where agent $2$ leaves at $t=15$: (a) the dynamics of $^f {e}_2^{(i)}$; (b) the dynamics of $^f {e}_3^{(i)}$.}
\label{fig:plugout}
\end{figure}

We further evaluate the proposed observer in scenarios involving agents joining or leaving.
Consider a homogeneous MAS with three agents, where the system matrices and observation matrices are identical to $A_{11}$ and $C_{11}$, and the coupling matrices are given as follows:
${C}_{21}={C}_{32}=\left[\begin{array}{cc}
1.2 & 0 \\
0 & 0.8 
\end{array}\right].$
The maximum allowable number of agents is $\overline{m}=4$.
The communication graph is also shown in Fig.~\ref{fig:communicationGcinout}(a).
In the agent-joining scenario where a new homogeneous agent (i.e., agent 4) joins the MAS at $t=15$, the corresponding communication graph is shown in Fig.~\ref{fig:communicationGcinout}(b). 
The measurements of agent $4$ are identical to those of agent $3$, i.e., $C_{44}=C_{33}$ and $C_{43}=C_{32}$.
In the agent-leaving scenario where agent $2$ leaves the MAS at $t=15$, the new communication graph is shown in Fig.~\ref{fig:communicationGcinout}(c).
In both scenarios, the Luenberger gains for all relevant agents are equal to $F_1$.
From the analysis in Subsection~\ref{sec33} and condition \eqref{eq:minmu1}, we obtain $\mu>571.43$.
We stipulate that each agent chooses the integer closest to this lower bound, setting $\mu=572$.
The consensus gains for each relevant agent are determined using \eqref{eq:consensusgain1} based on the updated communication graph when an agent joins or leaves.
Simulation results in Fig.~\ref{fig:plugin}-\ref{fig:plugout} demonstrate that the estimation error dynamics continue to converge to zero despite agents joining or leaving, illustrating the resilience of our design.
It is worth noting that all parameters of the proposed observer \eqref{eq:Do} can be independently determined by each agent throughout the process.

\subsection{Application to Cooperative Localization}\label{sec51}
Consider a multi-agent system consisting of six single-integrator agents.
The communication graph $\mathrm{G}_{\rm c}$ and sensing graph $\mathrm{G}_{\rm o}$ are shown in Fig.~\ref{fig:multirobot}, where agent $2$ can obtain the absolute position.
Based on $\mathrm{G}_{\rm o}$ and the definitions of the graphs $\mathrm{G}_{\rm o}^+$ and $\overline{\mathrm{G}}_{\rm o}^+$, it follows that each agent $i\in\mathcal V$ in $\mathrm{G}_{\rm o}^+$ is not a source, and $\overline{\mathrm{G}}_{\rm o}^+$ is connected.
Additionally, the communication graph $\mathrm{G}_{\rm c}$ in Fig.~\ref{fig:multirobot} is strongly connected.
Thus, all conditions of Theorem~\ref{pro:main} are satisfied.
Using the distributed DAG construction algorithm from Algorithm~\ref{alg:case2}, we obtain the directed acyclic sensing graph $\mathrm{G}_{\rm o}^*$, shown in Fig.~\ref{fig:go*}.
\begin{figure}[!ht]
\centerline{\includegraphics[width=5cm]{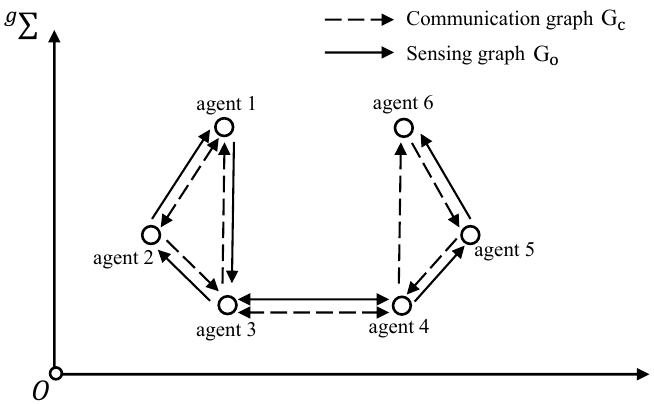}}
% \vspace{-10pt}
\caption{The graphs $\mathrm{G}_{\rm c}$ and $\mathrm{G}_{\rm o}$ in Subsection~\ref{sec51}.}
\label{fig:multirobot}
\end{figure}
% \vspace{-20pt}
\begin{figure}[!ht]
\centerline{\includegraphics[width=5cm]{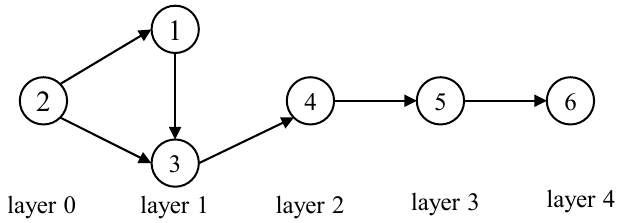}}
% \vspace{-7pt}
\caption{The directed acyclic sensing graph $\mathrm{G}_{\rm o}^*$ by using distributed DAG construction algorithm.}
% \vspace{-7pt}
\label{fig:go*}
\end{figure}

The control inputs for all agents are the same bounded signal: $u_i=[-0.1{\rm sin}(0.01t), 0.1{\rm cos}(0.01t)]^\top$ for all $i\in \{1, \ldots, 6\}$.
The initial positions are set as follows: $p_1(0)=[5,7]^\top, p_2(0)=[3,4]^\top, p_3(0)=[5,2]^\top, p_4(0)=[10,2]^\top, p_5(0)=[12,4]^\top, p_6(0)=[10,7]^\top$, with initial position estimates of all agents set to $\mathbf{0}$.
By Theorem~\ref{pro:main}, the Luenberger gains of \eqref{eq:localization} are chosen as
$F_1=F_4=F_5=F_6=\textup{diag}(-1,-0.5),
F_2=\left[\begin{array}{cccc}
    -1 & 0 & -1 & 0\\
    0 & -0.5 & 0 & -0.5
    \end{array}\right],~F_{3}=\textup{diag}(1,0.5),$
and all consensus gains $w_{i0}^{(i)}, w_{il}^{(j)}$ are set to $1$ for all $i,j\in \mathcal{V}, l\in \mathcal N_i(\mathrm G_{\rm c})$.
Fig.~\ref{fig:singleknown} shows the position estimation error dynamics of agents $3$ and $5$ by all agents under the condition that the input $u$ is known to all agents, where $^x{e}_j^{(i)}$ denotes the estimation error of the abscissa of the position of agent $j$ by agent $i$.
As depicted in Fig.~\ref{fig:singleknown}, all estimation error dynamics converge to zero, as expected from Theorem~\ref{pro:main}.
In the case where agents do not know the control inputs of others, Fig.~\ref{fig:singleunknown} demonstrates the true trajectory of agent $3$, the estimated trajectories of agent $3$ generated by all agents, and the corresponding position estimation error dynamics under bounded control inputs.
From Fig.~\ref{fig:singleunknown}, it is evident that all estimation error dynamics remain bounded, as expected from Subsection~\ref{sub:discussion_to_unknown_input}.

\begin{figure}[!ht]
\begin{minipage}{0.48\linewidth}
  \centerline{\includegraphics[width=4cm]{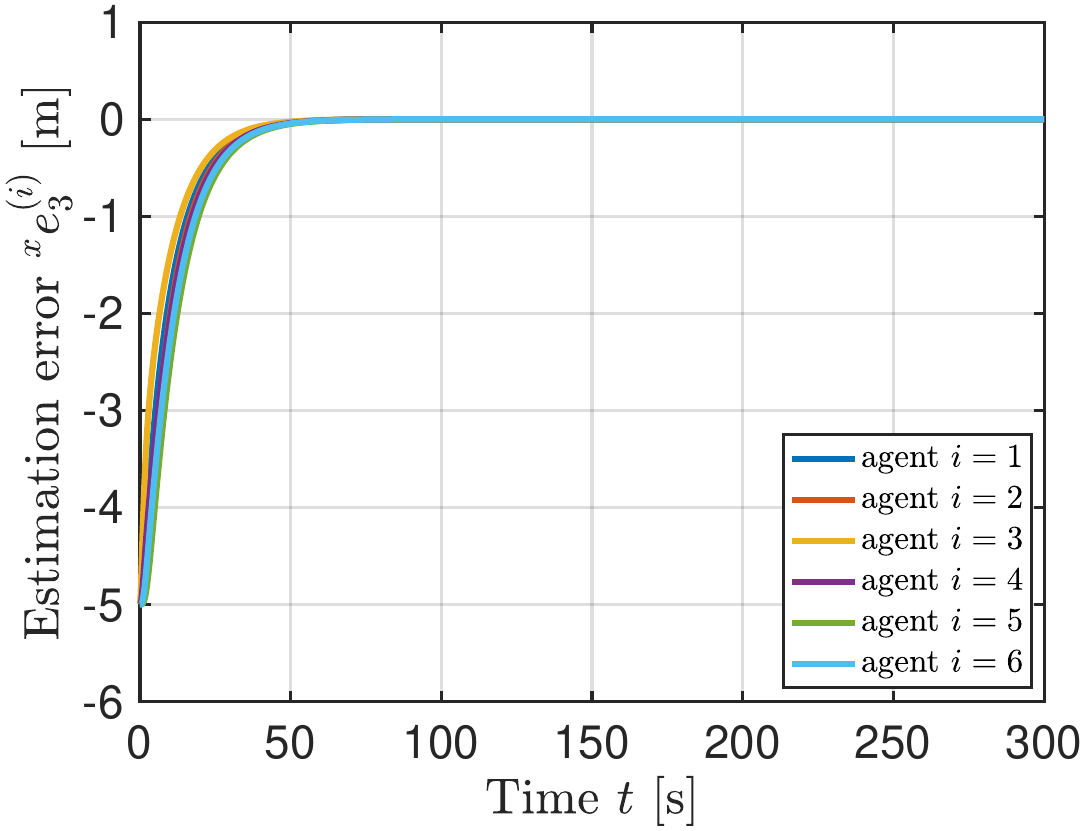}}
  \centerline{\scriptsize(a)}
\end{minipage}
\hfill
\begin{minipage}{0.48\linewidth}
  \centerline{\includegraphics[width=4cm]{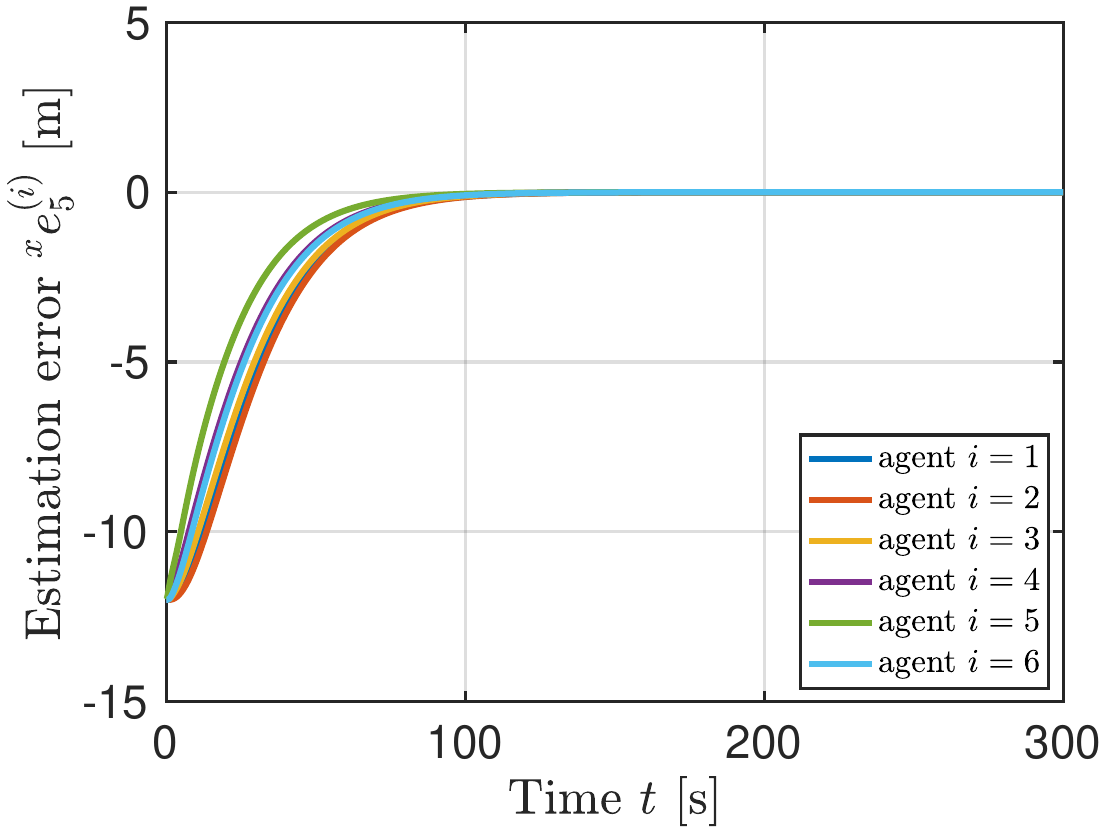}}
  \centerline{\scriptsize(b)}
\end{minipage}
\caption{Known input case: (a) the dynamics of $^x{e}_3^{(i)}$;
 (b) the dynamics of $^x{e}_5^{(i)}$.}
\label{fig:singleknown}
\end{figure}

\begin{figure}[!ht]
\begin{minipage}{0.48\linewidth}
  \centerline{\includegraphics[width=4cm]{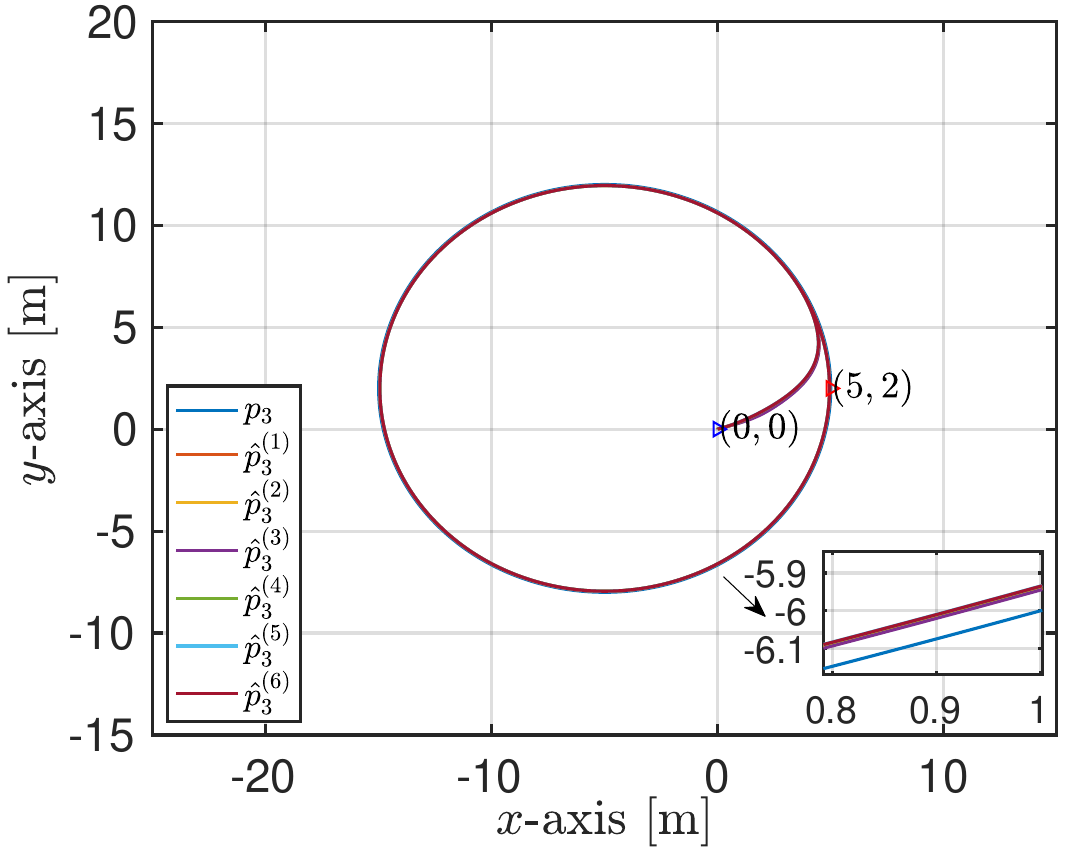}}
  \centerline{\scriptsize(a)}
\end{minipage}
\hfill
\begin{minipage}{0.48\linewidth}
  \centerline{\includegraphics[width=4.1cm]{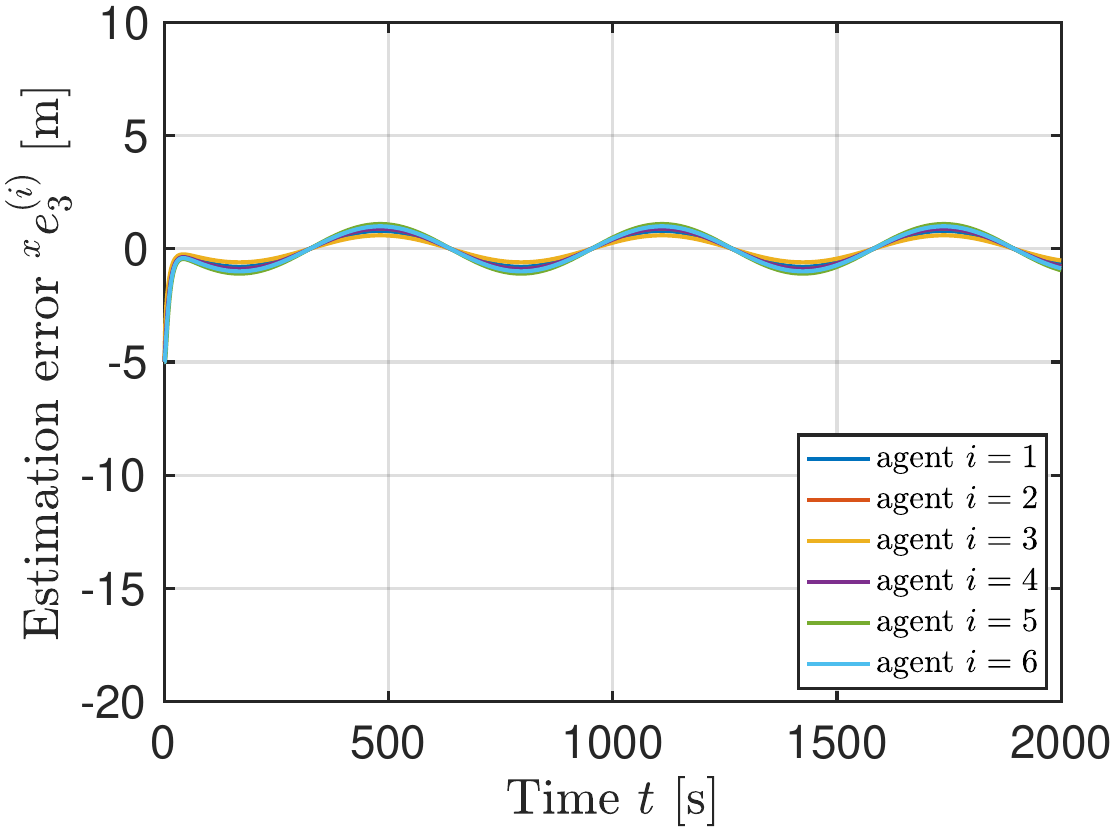}}
  \centerline{\scriptsize(b)}
\end{minipage}
\caption{Unknown input case: (a) real trajectory of agent $3$ and the estimated trajectories generated by all six agents;
(b) the dynamics of $^x{e}_3^{(i)}$.}
\label{fig:singleunknown}
\end{figure}

% \vspace{-5pt}
\section{Conclusions}\label{sec6}
In this paper, we presented a fully distributed observer for continuous-time linear MASs, with a focus on its application to cooperative localization.
We proved that the state estimate of the observer with a sufficiently large gain asymptotically converges to the state of the MAS.  
For a certain subclass of MASs, the proposed observer can be designed in a fully distributed manner, and it exhibits resilience to node changes of these MASs.
Building on our distributed estimation framework, we developed two fully distributed localization algorithms for single- and double-integrator MASs, respectively.
Different from the previous works, this paper considers the dynamics coupling among agents and the resilience of distributed estimation algorithms to node changes in MASs.
As future work, we aim to explore the design of a secure fully distributed state estimation algorithm for multi-agent systems in an adversarial environment, where some agents may behave maliciously.

\bibliographystyle{IEEEtran}      
\bibliography{IEEEabrv,distributed_observer}

\end{document}